\documentclass{article}

\usepackage[leqno,intlimits]{amsmath}
\usepackage{amssymb}
\usepackage{amsthm}
\usepackage{hyperref}
\usepackage{paralist}

\usepackage[margin=1in]{geometry}

\allowdisplaybreaks[1]

\newtheorem{theorem}{Theorem}[section]
\newtheorem{lemma}[theorem]{Lemma}

\newtheorem{claim}[theorem]{Claim}
\newtheorem{corollary}[theorem]{Corollary}
\newtheorem{definition}{Definition}  
\newtheorem{assumption}{Assumption}

\newtheoremstyle{example}{\topsep}{\topsep}%
  {}
  {}
  {\bfseries}
  {}
  {\newline}
	{\thmname{#1}\thmnumber{ #2}\thmnote{ #3}}

\theoremstyle{example}

\newcommand{\E}{\mathbb{E}} 
\newcommand{\p}{\mathbb{P}} 
\newcommand{\eps}{\varepsilon} 

\DeclareMathOperator{\diag}{diag}
\DeclareMathOperator{\Var}{Var}
\DeclareMathOperator{\Cov}{Cov}
\DeclareMathOperator{\Bin}{Bin}

\DeclareMathOperator*{\argmax}{arg\,max}
\DeclareMathOperator*{\argmin}{arg\,min}
\DeclareMathOperator*{\relint}{rel\,int}
\DeclareMathOperator{\sgn}{sgn}
\DeclareMathOperator{\aff}{aff}
\DeclareMathOperator{\cvx}{cvx}
\DeclareMathOperator{\cl}{cl}

\DeclareMathOperator{\R}{\mathbb{R}}
\DeclareMathOperator{\N}{\mathbb{N}}
\DeclareMathOperator{\Q}{\mathbb{Q}}
\DeclareMathOperator{\mcO}{\mathcal{O}}

\DeclareMathOperator{\GS}{GS}

\begin{document}

\title{A Smooth Transition from Powerlessness to Absolute Power}
\author{	Elchanan Mossel
	\thanks{University of California, Berkeley and Weizmann Institute of Science; \texttt{mossel@stat.berkeley.edu}; supported by NSF (DMS 1106999) and by DOD ONR grant N000141110140.}
	\and
	Ariel D. Procaccia
	\thanks{Carnegie Mellon University; \texttt{arielpro@cs.cmu.edu}.}
	\and
	Mikl\'os Z. R\'acz
	\thanks{University of California, Berkeley; \texttt{racz@stat.berkeley.edu}; supported by a UC Berkeley Graduate Fellowship and by NSF (DMS 0548249).}
}
\date{\today}
\maketitle


\begin{abstract}
We study the phase transition of the coalitional manipulation problem for generalized scoring rules. Previously it has been shown that, under some conditions on the distribution of votes, if the number of manipulators is $o \left( \sqrt{n} \right)$, where $n$ is the number of voters, then the probability that a random profile is manipulable by the coalition goes to zero as the number of voters goes to infinity, whereas if the number of manipulators is $\omega \left( \sqrt{n} \right)$, then the probability that a random profile is manipulable goes to one. 
Here we consider the critical window, where a coalition has size $c \sqrt{n}$, and we show that as $c$ goes from zero to infinity, the limiting probability that a random profile is manipulable goes from zero to one in a smooth fashion, i.e., there is a smooth phase transition between the two regimes. This result analytically validates recent empirical results, and suggests that deciding the coalitional manipulation problem may be of limited computational hardness in practice.
\end{abstract}


\section{Introduction} 

Finding ``good'' voting systems which satisfy some natural requirements is one of the main goals in social choice theory. This problem is increasingly relevant in the area of artificial intelligence and in computer science more broadly, where virtual elections are now an established tool for preference aggregation (see, e.g.,~\cite{CP11}).

A naturally desirable property of a voting system is \emph{strategyproofness} (a.k.a.\ nonmanipulability): no voter should benefit from voting strategically, i.e., voting not according to her true preferences. However, Gibbard~\cite{gibbard1973manipulation} and Satterthwaite~\cite{satterthwaite1975strategy} showed that no reasonable voting system can be strategyproof. Before stating their result, let us specify the problem more formally. 

We consider $n$ voters electing a single winner among $m$ candidates. The voters specify their opinion by ranking the candidates, and the winner is determined according to some predefined \emph{social choice function} (SCF) $f : S_m^n \to \left[m\right]$ of all the voters' rankings, where $S_m$ denotes the set of all possible total orderings of the $m$ candidates. We call a collection of rankings by the voters a \emph{ranking profile}. We say that a SCF is \emph{manipulable} if there exists a ranking profile where a voter can achieve a more desirable outcome of the election according to her true preferences by voting in a way that does not reflect her true preferences.

The Gibbard-Satterthwaite theorem states that any SCF which is not a dictatorship (i.e., not a function of a single voter), and which allows at least three candidates to be elected, is manipulable. This has contributed to the realization that it is unlikely to expect truthfulness in voting. Consequently, there have been many branches of research devoted to understanding the extent of the manipulability of voting systems, and to finding ways of circumventing the negative results.

One approach, introduced by Bartholdi, Tovey, and Trick~\cite{bartholdi1989computational}, suggests computational complexity as a barrier against manipulation: a SCF may not be manipulable in practice if it is hard for a voter to compute a manipulative vote. 
A significant body of work focuses on the worst-case complexity of manipulation (see the survey by Faliszewski and Procaccia~\cite{faliszewski2010ai}). 
Here we are interested specifically in the \emph{coalitional manipulation} problem, where a group of voters can change their votes in unison, with the goal of making a given candidate win. 
Various variations of this problem are known to be $\mathcal{NP}$-hard under many of the common SCFs~\cite{CSL07,XZPC+09,BNW11}.

Crucially, this line of work focuses on worst-case complexity. 
While worst-case hardness of manipulation is a desirable property for a SCF to have, it does not tell us much about \emph{typical} instances of the problem---is it \emph{usually} easy or hard to manipulate? 
A recent line of research on average-case manipulability has been questioning the validity of such worst-case complexity results. 
The goal of this alternative line of work is to show that there are no ``reasonable'' voting rules that are computationally hard to manipulate \emph{on average}. 
Specifically, the goal is to rule out the following informal statement: there are ``good'' voting rules that are hard to manipulate on average under any ``sufficiently rich'' distribution over votes. 

Taking this point of view, showing easiness of manipulation under a restricted class of distributions---such as i.i.d.\ votes or even uniform votes (the impartial culture assumption)---is interesting, even if these do not necessarily capture all possible real-world elections. 
Specifically, if we show that manipulation is easy under such distributions, then any average-case hardness result would necessarily have to make some unnatural technical assumptions to avoid these distributions. 
Studying such restricted distributions over votes is indeed exactly what some recent papers have done.

For the coalitional manipulation problem, 
Procaccia and Rosenschein~\cite{procaccia2007average} first suggested that it is trivial to determine whether manipulation is possible for most coalitional manipulation instances, from a typical-case computational point of view; one can make a highly informed guess purely based on the number of manipulators. 
Specifically, they studied a setting where there is a distribution over votes (which satisfies some conditions), and concentrated on a family of SCFs known as positional scoring rules. 
They showed that if the size of the coalition is $o \left( \sqrt{n} \right)$, then with probability converging to 1 as $n \to \infty$, the coalition is powerless, i.e., it cannot change the outcome of the election. 
In contrast, if the size of the coalition is $\omega \left( \sqrt{n} \right)$ (and $o \left( n \right)$), then with probability converging to 1 as $n \to \infty$, the coalition is all-powerful, i.e., it can elect any candidate. 
Later Xia and Conitzer~\cite{xia2008generalized} proved an analogous result for so-called generalized scoring rules, a family that contains almost all common voting rules. 
See also related work by Peleg~\cite{peleg1979note}, Slinko~\cite{Slin04}, Pritchard and Slinko~\cite{pritchard2006average}, and Pritchard and Wilson~\cite{PW09}. We discuss additional related work in Section~\ref{sec:related}.

Our primary interest in this paper is to understand the critical window that these papers leave open, when the size of the coalition is $\Theta \left( \sqrt{n} \right)$. 
Specifically, we are interested in the \emph{phase transition} in the probability of coalitional manipulation, when the size of the coalition is $c \sqrt{n}$ and $c$ varies from zero to infinity, i.e., the \emph{transition from powerlessness to absolute power}. 

In the past few decades there has been much research on the connection between phase transitions and computationally hard problems (see, e.g.,~\cite{fu1986application,cheeseman1991really,achlioptas2005rigorous}). 
In particular, it is often the case that the computationally hardest problems can be found at critical values of a sharp phase transition (see, e.g.,~\cite{gomes2006randomness} for an overview). 
On the other hand, smooth phase transitions are often found in connection with computationally easy (polynomial) problems, such as 2-coloring~\cite{achlioptas1999threshold} and 1-in-2 SAT~\cite{walsh2002interface}. 
Thus understanding the phase transition in this critical window may shed light on where the computationally hardest problems lie.

Recently, Walsh~\cite{walsh2011really} \emph{empirically} analyzed two well-known voting rules---veto and single transferable vote (STV)---and found that there is a smooth phase transition between the two regimes. Specifically, Walsh studied coalitional manipulation with unweighted votes for STV and weighted votes for veto, and sampled from a number of distributions in his experiments, including i.i.d.\ distributions, correlated distributions, and votes sampled from real-world elections. Our main result complements and improves upon Walsh's analysis in two ways; while Walsh's results show how the phase transition looks like concretely for veto and STV, we \emph{analytically} show that the phase transition is indeed smooth for \emph{any generalized scoring rule} (including veto and STV) when the votes are i.i.d. This suggests that deciding the coalitional manipulation problem may not be computationally hard in practice.

\subsection{Our results} 

We now present our results, but first let us formally specify the setup of the problem. We denote a ranking profile by $\sigma = \left( \sigma_1, \dots, \sigma_n \right) \in S_m^n$, and for a candidate $a$, define $W_a = \left\{ \sigma \in S_m^n \, \middle| \, f\left( \sigma \right) = a \right\}$, the set of ranking profiles where the outcome of $f$ is $a$. Our setup and assumptions are the following.
\begin{assumption}\label{ass:1}
 We assume that the number of candidates, $m$, is constant.
\end{assumption}
\begin{assumption}\label{ass:2}
 We assume that the SCF $f$ is \emph{anonymous}, i.e., it treats each voter equally.
\end{assumption}
\begin{assumption}\label{ass:3}
 We assume that the votes of voters are i.i.d., according to some distribution $p$ on $S_m$. Furthermore, we assume that there exists $\delta > 0$ such that for every $\pi \in S_m$, $p \left( \pi \right) \geq \delta$ (necessarily $\delta \leq 1/ m!$).
\end{assumption}
If we were to assume only these, then our setup would include uninteresting cases, such as when $f$ is a constant---i.e., no matter what the votes are, a specific candidate wins. 
Another less interesting case is when the probability of a given candidate winning vanishes as $n \to \infty$---we can then essentially forget about this candidate for large $n$ (in the sense that a coalition of size $\Omega \left( n \right)$ would be necessary to make this candidate win). 
To exclude these and focus on the interesting cases, we make an additional assumption which concerns both the SCF and the distribution of votes.
\begin{assumption}\label{ass:4}
 We assume that there exists $\eps > 0$ such that for every $n$ and for every candidate $a \in \left[m\right]$, the probability of $a$ being elected is at least $\eps > 0$, i.e., $\p \left( W_a \right) \geq \eps$ (necessarily $\eps \leq 1 / m$).
\end{assumption}
All four assumptions are satisfied when the distribution is uniform (i.e., under the impartial culture assumption) and the SCF is close to being neutral (i.e., neutral up to some tie-breaking rules); in particular, they hold for all commonly used SCFs.
The assumptions are somewhat more general than this, although the i.i.d.\ assumption remains a restrictive one. 
However, as discussed earlier, even showing easiness of manipulation under such a restricted class of distributions is interesting. 

As mentioned before, we are interested in the case when the coalition has size $c \sqrt{n}$ for some constant $c$. Define the probabilities
\begin{align*}
 \underline{q}_n \left( c \right) &:= \p \left( \text{ some coalition of size } c \sqrt{n} \text{ can elect any candidate } \right),\\
 \overline{q}_n \left( c \right) &:= \p \left( \text{ some coalition of size } c \sqrt{n} \text{ can change the outcome of the election } \right),\\
 \underline{r}_n \left( c \right) &:= \p \left( \text{ a specific coalition of size } c \sqrt{n} \text{ can elect any candidate } \right),\\
 \overline{r}_n \left( c \right) &:= \p \left( \text{ a specific coalition of size } c \sqrt{n} \text{ can change the outcome of the election } \right),
\end{align*}
and let
\[
 \underline{q} \left( c \right) := \lim_{n \to \infty} \underline{q}_n \left( c \right), \quad \overline{q} \left( c \right) := \lim_{n \to \infty} \overline{q}_n \left( c \right), \quad \underline{r} \left( c \right) := \lim_{n \to \infty} \underline{r}_n \left( c \right), \quad \overline{r} \left( c \right) := \lim_{n \to \infty} \overline{r}_n \left( c \right),
\]
provided these limits exist. Clearly $\underline{q}_n \left( c \right) \leq \overline{q}_n \left( c \right)$, $\underline{r}_n \left( c \right) \leq \overline{r}_n \left( c \right)$, $\underline{r}_n \left( c \right) \leq \underline{q}_n \left( c \right)$, and $\overline{r}_n \left( c \right) \leq \overline{q}_n \left( c \right)$. 

Before we describe our results, which deal with these quantities, we first explain how these relate to the various variants of the coalitional manipulation problem. 
In the coalitional manipulation problem the coalition is fixed, and thus the relevant quantities are $\underline{r}_n \left( c \right)$ and $\overline{r}_n \left( c \right)$. 
Closely related is the problem of determining the \emph{margin of victory}, which is the minimum number of voters who need to change their votes to change the outcome of the election. 
Also related is the problem of \emph{bribery}, the minimum number of voters who need to change their votes to make a given candidate win. 
The main difference between these problems is that in coalitional manipulation the coalition is fixed, whereas in the latter two problems the coalition is not fixed. 
Hence the relevant quantities for studying the latter two are $\underline{q}_n \left( c \right)$ and $\overline{q}_n \left( c \right)$. 
Our tools also allow us to deal with other related quantities (such as \emph{microbribery}~\cite{faliszewski2009llull}), but we focus our attention on the four quantities described above.

Our first result analyzes the case when the size of the coalition is $c \sqrt{n}$ for large $c$. 
We show that if $c$ is large enough, then with probability close to 1, a specific coalition of size $c \sqrt{n}$ can elect any candidate. 
This holds for any SCF that satisfies the above (mild) restrictions.

\begin{theorem}\label{thm:main1}
 Assume that Assumptions~\ref{ass:1},~\ref{ass:2},~\ref{ass:3}, and~\ref{ass:4} hold. 
For any $\eta > 0$ there exists a constant $c = c\left( \eta, \delta, \eps, m \right)$ such that $\underline{r}_n \left( c \right) \geq 1 - \eta$ for every $n$. 
In particular, we can choose 
$$
c = \left( 4 / \delta \right) \log \left( 2 m! / \eta \right) \left[ \sqrt{\log \left( 2m / \eta \right)} + \sqrt{\log \left( 2 / \eps \right)} \right].
$$ 
It follows that
\[
 \lim_{c \to \infty} \liminf_{n} \underline{r}_n \left( c \right) = 1.
\]
\end{theorem}
This result extends theorems of Procaccia and Rosenschein~\cite{procaccia2007average}, and Xia and Conitzer~\cite{xia2008generalized}, from scoring rules and generalized scoring rules, respectively, to anonymous SCFs.

Our second result deals with the case when the size of the coalition is $c \sqrt{n}$ for small $c$, and the transition as $c$ goes from 0 to $\infty$. 
Here we assume additionally that $f$ is a generalized scoring rule (to be defined in Section~\ref{sec:gen_sc_r}); this is needed because there exist (pathological) anonymous SCFs for which the result below does not hold (see the beginning of Section~\ref{sec:small} for an example).

\begin{theorem}\label{thm:main2}
 Assume that Assumptions~\ref{ass:1},~\ref{ass:2},~\ref{ass:3}, and~\ref{ass:4} hold, and furthermore that $f$ is a generalized scoring rule. Then:
    \begin{enumerate}[(1)]
      \item\label{limit} The limits $\underline{q} \left( c \right)$, $\overline{q} \left( c \right)$, $\underline{r} \left( c \right)$ and $\overline{r} \left( c \right)$ exist.
      \item\label{small} There exists a constant $K = K \left( f, \delta \right) < \infty$ such that $\overline{q} \left( c \right) \leq K c$; in particular, $\lim_{c \to 0} \overline{q} \left( c \right) = 0$. 
      \item\label{smooth} For all $0 < c < \infty$, $0 < \underline{q} \left( c \right) \leq \overline{q} \left( c \right) < 1$ and $0 < \underline{r} \left( c \right) \leq \overline{r} \left( c \right) < 1$, and furthermore $\underline{q} \left( c \right)$, $\overline{q} \left( c \right)$, $\underline{r} \left( c \right)$ and $\overline{r} \left( c \right)$ are all continuously differentiable in $c$ with bounded derivative. 
    \end{enumerate}
\end{theorem}

In words, Part~\ref{small} means that if $c$ is small enough then with probability close to 1 no coalition of size $c \sqrt{n}$ can change the outcome of the election, and the statements about $\overline{r}$ and $\underline{r}$ in Part~\ref{smooth} mean that the coalitional manipulation problem has a smooth phase transition: as the number of manipulators increases, the probabilities that a coalition has some power, and that it has absolute power, increase smoothly. Parts \ref{limit} and \ref{small} of the theorem simply make a result of Xia and Conitzer~\cite{xia2008generalized} more precise, by extending the analysis to the $\Theta(\sqrt{n})$ regime. More importantly, in the proofs of these statements we introduce the machinery needed to establish Part~\ref{smooth}, which is our main result. 

Since the coalitional manipulation problem does \emph{not} have a sharp phase transition, Theorem~\ref{thm:main2} can be interpreted as suggesting that realistic distributions over votes are likely to yield coalitional manipulation instances that are tractable in practice, even if the size of the coalition concentrates on the previously elusive $\Theta(\sqrt{n})$ regime; this is true for any generalized scoring rule, and in particular for almost all common social choice functions (an exception is Dodgson's rule). 
This interpretation has a negative flavor in further strengthening the conclusion that worst-case complexity is a poor barrier to manipulation. 

However, the complexity glass is in fact only half empty. The probability that the margin of victory is at most $c\sqrt{n}$ is captured by the quantity $\overline{q}_n$, hence Part~\ref{smooth} of Theorem~\ref{thm:main2} also implies that the margin of victory problem has a smooth phase transition. As recently pointed out by Xia~\cite{xia2012computing}, efficiently solving the margin of victory problem could help in \emph{post-election audits}---used to determine whether electronic elections have resulted in an incorrect outcome due to software or hardware bugs---and its tractability is in fact desirable.

The methods we use are flexible, and can be extended to various setups of interest that do not directly satisfy our assumptions above, for instance single-peaked preferences. 
Consider a one-dimensional political spectrum represented by the interval $\left[0,1\right]$, and fix the location of the candidates. 
Assume voters are uniformly distributed on the interval, independently of each other. 
For technical reasons, this distribution does not satisfy our assumptions, since there will be rankings $\pi \in S_m$ such that $p\left( \pi \right) = 0$; however, our tools allow us to deal with this setting as well. 
For instance, if the locations of the $m$ candidates are $\left\{ \frac{1}{2m}, \frac{3}{2m}, \dots, \frac{2m-1}{2m} \right\}$, then our results hold (with appropriate quantitative modifications). 
Similarly, if the locations were something else, then there would exist a subset of candidates who have an asymptotically nonvanishing probability of winning, and the same results hold restricted to this subset of candidates.

Finally, we discuss the role of tie-breaking in our setup, since this is often an important issue when studying manipulation. 
However, since we consider manipulation by coalitions of size $c \sqrt{n}$, ties where there exist a constant number of voters such that if their votes are changed appropriately there is no longer a tie, are not relevant. 
Indeed, our tools allow us to extend the results of Theorem~\ref{thm:main2} to a class of SCFs slightly beyond generalized scoring rules, and, in particular, these allow for arbitrary tie-breaking rules (see Section~\ref{sec:almost} for details).

\subsection{Additional related work}\label{sec:related} 

A recent line of research with an average-case algorithmic flavor also suggests that manipulation is indeed typically easy; see, e.g., the work of \ Kelly~\cite{kelly1993almost}, Conitzer and Sandholm~\cite{conitzer2006nonexistence}, Procaccia and Rosenschein~\cite{procaccia2007junta}, and Zuckerman et al.~\cite{ZPR09} for results on certain restricted classes of SCFs. A different approach, initiated by Friedgut, Kalai, Keller and Nisan~\cite{friedgut2008elections,friedgut2011quantitative}, who studied the fraction of ranking profiles that are manipulable, also suggests that manipulation is easy on average; see further Xia and Conitzer~\cite{xia2008sufficient}, Dobzinski and Procaccia~\cite{dobzinski2008frequent}, Isaksson, Kindler and Mossel~\cite{isaksson2010geometry}, and Mossel and R\'acz~\cite{mossel2011quantitative}. We refer to the survey by Faliszewski and Procaccia~\cite{faliszewski2010ai} for a detailed history of the surrounding literature. See also related literature in economics, e.g.,~\cite{
good1975estimating,chamberlain1981note,myatt2007theory}.

Recent work by Xia~\cite{xia2012computing} is independent from, and closely related to, our work. As mentioned above, Xia's paper is concerned with computing the margin of victory in elections. He focuses on computational complexity questions and approximation algorithms, but one of his results is similar to Parts~\ref{limit} and \ref{small} of Theorem~\ref{thm:main2}. However, our analysis is completely different; our approach facilitates the proof of Part~\ref{smooth} of the theorem, which is our main contribution. An even more recent (and also independent) manuscript by Xia~\cite{xia2012many} considers similar questions for generalized scoring rules and captures additional types of strategic behavior (such as control), but again, crucially, this work does not attempt to understand the phase transition (nor does it subsume our Theorem~\ref{thm:main1}).

\section{Large coalitions}\label{sec:large} 

Without further ado, we prove Theorem~\ref{thm:main1}. 
The main idea is to observe that for i.i.d.\ distributions, the Hamming distance of a random ranking profile from a fixed subset of ranking profiles concentrates around its mean. 
The theorem follows from standard concentration inequalities. 

\begin{proof}[Proof of Theorem~\ref{thm:main1}] 
For $\sigma, \sigma' \in S_m^n$, define
\[
 d\left( \sigma, \sigma' \right) = \frac{1}{n} \sum_{i=1}^{n} \mathbf{1} \left[ \sigma_i \neq \sigma'_i \right],
\]
i.e., $d\left( \sigma, \sigma' \right)$ is $1/n$ times the Hamming distance of $\sigma$ and $\sigma'$. If $U$ is a subset of ranking profiles and $\sigma$ is a specific ranking profile then define $d_{U} \left( \sigma \right) = \min_{\sigma' \in U} d\left( \sigma, \sigma' \right)$. The function $d_U$ is Lipschitz with constant $1/n$, and therefore by McDiarmid's inequality we have the following concentration inequality:
\begin{equation}\label{eq:conc_d_U}
  \p \left( \left| d_U \left( \sigma \right) - \E d_U \right| \geq c \right) \leq 2 \exp \left( -2c^2 n \right)
\end{equation}
for any $c > 0$ and $U \subseteq S_m^n$. Suppose $U \subseteq S_m^n$ has measure at least $\eps$, i.e., $U$ is such that $\p \left( \sigma \in U \right) \geq \eps$, and take $\gamma$ such that $2 \exp \left( - 2 \gamma^2 n \right) < \eps$, e.g., let $\gamma = \sqrt{\log \left( 2/ \eps \right)} / \sqrt{n}$. Then \eqref{eq:conc_d_U} implies that there exists $\sigma \in U$ such that $\left| d_U \left( \sigma \right) - \E d_U \right| \leq \gamma$, but since $d_U \left( \sigma \right) = 0$, this means that $\E d_U \leq \gamma$. So for such a set $U$, we have
\[
 \p \left( d_U \left( \sigma \right) > \gamma + c \right) \leq \exp \left( - 2 c^2 n \right)
\]
for any $c > 0$. Choosing $c = B / \sqrt{n}$ and defining $B' =  B + \sqrt{\log \left( 2/ \eps \right)}$ we get that
\begin{equation}\label{eq:conc2}
 \p \left( d_U \left( \sigma \right) > B' / \sqrt{n} \right) \leq \exp \left( - 2 B^2 \right).
\end{equation}
In the language of the usual Hamming distance, this means that the probability that the ranking profile needs to be changed in at least $B'\sqrt{n}$ coordinates to be in $U$ is at most $\exp \left( - 2 B^2 \right)$, which can be made arbitrarily small by choosing $B$ large enough.

By our assumption, $\p \left( \sigma \in W_a \right) \geq \eps$ for every $a$, and therefore by \eqref{eq:conc2} and a union bound we get
\[
 \p \left( \exists a : d_{W_a} \left( \sigma \right) > B' / \sqrt{n} \right) \leq m \exp \left( - 2 B^2 \right).
\]
By choosing $B = \sqrt{\log\left( 2m / \eta \right)}$, this probability is at most $\eta / 2$.

Consider a specific coalition of size $D B' \sqrt{n}$, where $D = D\left( \delta, m \right)$ will be chosen later. Using Chernoff's bound and a union bound, with probability close to one, for every possible ranking $\pi$ the coalition has at least $B' \sqrt{n}$ voters with the ranking $\pi$:
\begin{multline*}
 \p\left( \exists \pi \in S_m  : \text{ coalition of size } D B' \sqrt{n} \text{ has less than } B'\sqrt{n} \text{ voters with ranking } \pi \right)\\
\begin{aligned}
 &\leq m! \p \left( \Bin \left( DB' \sqrt{n}, \delta \right) < B' \sqrt{n} \right)
 \leq m! \exp \left( -  \left( 1 - 1 / D \delta \right)^2 DB' \sqrt{n} \delta / 2 \right) \\
 &\leq m! \exp \left( - \left( 1 - 1/ D \delta \right)^2 D \delta / 2 \right),
\end{aligned}
\end{multline*}
where $\Bin \left( DB' \sqrt{n}, \delta \right)$ denotes a binomial random variable with parameters $DB' \sqrt{n}$ and $\delta$, and where we used our assumption that for every voter the probability for every ranking is at least $\delta > 0$. Choosing $D = \left( 4 / \delta \right) \log \left( 2 m! / \eta \right)$, this probability is at most $\eta / 2$.

By the anonymity of $f$, the outcome only depends on the number of voters voting according to each ranking. Consequently, if $\sigma$ is such that it is at a distance of at most $B'/\sqrt{n}$ away from each $W_a$, and where for each ranking $\pi$ there are at least $B'\sqrt{n}$ voters in the coalition with ranking $\pi$, then the coalition is able to achieve any outcome. Using the above and a union bound this happens with probability at least $1 - \eta$. 
\end{proof}

\section{Small coalitions and the phase transition}\label{sec:small} 

This section is almost entirely devoted to the proof of Theorem~\ref{thm:main2}, but it also includes some helpful definitions, examples, and intuitions. 

Consider the following example of a SCF. 
For $a \in \left[m\right]$ let $n_a \left( \sigma \right)$ denote the number of voters who ranked candidate $a$ on top in the ranking profile $\sigma$. 
Define the SCF $f$ by $f \left( \sigma \right) = \sum_{a=1}^{m} a n_a \left( \sigma \right) \mod m$. 
This SCF is clearly anonymous (since it only depends on the number of voters voting according to specific rankings), and moreover it is easy to see that \emph{any} single voter can \emph{always} elect any candidate. 

This example shows that, in general, we cannot have a matching lower bound for the size of the manipulating coalition on the order of $\sqrt{n}$. However, this is an artificial example (one would not consider such a voting system in real life), and we expect that a matching lower bound holds for most reasonable SCFs.

Xia and Conitzer~\cite{xia2008generalized} introduced a large class of SCFs called generalized scoring rules, which include most commonly occurring SCFs. In the following we introduce this class of SCFs, provide an alternative way of looking at them (as so-called ``hyperplane rules''), and show that for this class of SCFs if the coalition has size $c \sqrt{n}$ for small enough $c$, then the probability of being able to change the outcome of the election can be arbitrarily close to zero. At the end of the section we then prove the smooth transition as stated in Part~\ref{smooth} of Theorem~\ref{thm:main2}.

\subsection{Generalized scoring rules, hyperplane rules, and their equivalence}\label{sec:GSR_hyp_rules} 

\subsubsection{Generalized scoring rules}\label{sec:gen_sc_r} 

We now define generalized scoring rules.
\begin{definition}\label{def:equiv}
 For any $y,z \in \R^k$, we say that $y$ and $z$ are \emph{equivalent}, denoted by $y \sim z$, if for every $i,j \in \left[ k \right]$, $y_i \geq y_j$ if and only if $z_i \geq z_j$.
\end{definition}
\begin{definition}
 A function $g : \R^k \to \left[m \right]$ is \emph{compatible} if for any $y \sim z$, $g \left( y \right) = g \left( z \right)$.
\end{definition}
That is, for any function $g$ that is compatible, $g \left( y \right)$ is completely determined by the total preorder of $\left\{ y_1, \dots, y_k \right\}$ (a total preorder is an ordering in which ties are allowed).
\begin{definition}[Generalized scoring rules]
 Let $k \in \N$, $f : S_m \to \R^k$ (called a \emph{generalized scoring function}), and $g : \R^k \to \left[ m \right]$ where $g$ is compatible ($g$ is called a \emph{decision function}). The functions $f$ and $g$ determine the \emph{generalized scoring rule} $\GS \left( f, g \right)$ as follows: for $\sigma \in S_m^n$, let
\[
 \GS \left( f, g \right) \left( \sigma \right) := g \left( \sum_{i=1}^n f \left( \sigma_i \right) \right).
\]
\end{definition}
From the definition it is clear that every generalized scoring rule (GSR) is anonymous.

\subsubsection{Hyperplane rules}\label{sec:hyp_rules} 

{\bf Preliminaries and notation.} In the following, for a SCF let us write $f \equiv f_n$, i.e., let us explicitly note that $f$ is a function on $n$ voters; also let us write $\sigma \equiv \sigma^n$. Since the SCF $f_n$ is anonymous, the outcome only depends on the numbers of voters who vote according to particular rankings.  Let $D_n$ denote the set of points in the probability simplex $\Delta^{m!}$ for which all coordinates are integer multiples of $1/n$. Let us denote a typical element of the probability simplex $\Delta^{m!}$ by $x = \left\{ x_{\pi} \right\}_{\pi \in S_m}$. For a ranking profile $\sigma^n$, let us denote the corresponding element of the probability simplex by $x\left( \sigma^n \right)$, i.e., for all $\pi \in S_m$,
\[
 x\left( \sigma^n \right)_{\pi} = \frac{1}{n} \sum_{i=1}^n \mathbf{1} \left[ \sigma_i = \pi \right].
\]
By our assumptions the outcome of $f_n$ only depends on $x\left( \sigma^n \right)$, so by abuse of notation we may write that $f_n : \Delta^{m!}|_{D_n} \to \left[ m \right]$ with $f_n \left( \sigma^n \right) = f_n \left( x\left( \sigma^n \right) \right)$.

We are now ready to define hyperplane rules.

\begin{definition}[Hyperplane rules]
 Fix a finite set of affine hyperplanes of the simplex $\Delta^{m!}$: $H_1, \dots, H_\ell$. Each affine hyperplane partitions the simplex into three parts: the affine hyperplane itself and two open halfspaces on either side of the affine hyperplane. Thus the affine hyperplanes $H_1, \dots, H_{\ell}$ partition the simplex into finitely many (at most $3^{\ell}$) regions. Let $F : \Delta^{m!} \to \left[m \right]$ be a function which is constant on each such region. Then the sequence of SCFs $\left\{ f_n \right\}_{n \geq 1}$, $f_n : S_m^n \to \left[ m \right]$, defined by
\[
 f_n \left( \sigma^n \right) = F \left( x \left( \sigma^n \right) \right)
\]
is called a hyperplane rule induced by the affine hyperplanes $H_1, \dots, H_{\ell}$ and the function $F$.
\end{definition}

A function $F : \Delta^{m!} \to \left[m \right]$ naturally partitions the simplex $\Delta^{m!}$ into $m$ parts based on the outcome of $F$. (For hyperplane rules this partition is coarser than the partition of $\Delta^{m!}$ induced by the affine hyperplanes $H_1, \dots, H_{\ell}$.) We abuse notation and denote these parts by $\left\{W_a \right\}_{a \in \left[m\right]}$. The following definition will be useful for us.

\begin{definition}[Interior and boundaries of a partition induced by $F$]
 We say that $x \in \Delta^{m!}$ is an interior point of the partition $\left\{W_a \right\}_{a \in \left[m\right]}$ induced by $F$ if there exists $\alpha > 0$ such that for all $y \in \Delta^{m!}$ for which $\left| x - y \right| \leq \alpha$, we have $F\left( x \right) = F \left( y \right)$. Otherwise, we say that $x \in \Delta^{m!}$ is on the boundary of the partition, which we denote by $B$.
\end{definition}

For a hyperplane rule the boundary $B$ is contained in the union of the corresponding affine hyperplanes. Conversely, suppose $F : \Delta^{m!} \to \left[m \right]$ is an arbitrary function and the sequence of (anonymous) SCFs $\left\{ f_n \right\}_{n \geq 1}$, $f_n : S_m^n \to \left[ m \right]$ is defined by $f_n \left( \sigma^n \right) = F \left( x \left( \sigma^n \right) \right)$. If the boundary $B$ of $F$ is contained in the union of finitely many affine hyperplanes of $\Delta^{m!}$, then $F$ is not necessarily a hyperplane rule, but there exists a hyperplane rule $\hat{F}$ such that $F$ and $\hat{F}$ agree everywhere except perhaps on the union of the finitely many affine hyperplanes.

\subsubsection{Equivalence}\label{sec:equiv} 

Xia and Conitzer~\cite{xia2009finite} gave a characterization of generalized scoring rules: a SCF is a generalized scoring rule if and only if it is anonymous and finitely locally consistent (see Xia and Conitzer~\cite[Definition 5]{xia2009finite}). This characterization is related to saying that generalized scoring rules are the same as hyperplane rules, yet we believe that spelling this out explicitly is important, because the geometric viewpoint of hyperplane rules is somewhat different, and in this probabilistic context it is also more flexible.

\begin{lemma}\label{lem:equiv}
 The class of generalized scoring rules coincides with the class of hyperplane rules.
\end{lemma}

\begin{proof}
 First let us show that every hyperplane rule is a generalized scoring rule. Let us consider the hyperplane rule induced by affine hyperplanes $H_1, \dots, H_{\ell}$ of the simplex $\Delta^{m!}$, and the function $F : \Delta^{m!} \to \left[m \right]$. The affine hyperplanes of $\Delta^{m!}$ can be thought of as hyperplanes of $\R^{m!}$ that go through the origin---abusing notation we also denote these by $H_1, \dots, H_{\ell}$. Let $u_1, \dots, u_{\ell}$ denote unit normal vectors of these hyperplanes.

 We need to define functions $f$ and $g$ such that for every ranking profile $\sigma^n \in S_m^n$, $\GS \left( f,g \right) \left( \sigma^n \right) = F \left( x \left( \sigma^n \right) \right)$. We will have $f : S_m \to \R^{\ell + 1}$ and $g : \R^{\ell + 1} \to \left[ m \right]$. Coordinates $1, \dots, \ell$ of $f$ correspond to hyperplanes $H_1, \dots, H_{\ell}$, while the last coordinate of $f$ will always be 0 (this is a technical necessity to make sure that the function $g$ is compatible). Let us look at the coordinate corresponding to hyperplane $H_j$ with normal vector $u_j$. For $\pi \in S_m$ define
\[
 \left( f \left( \pi \right) \right)_j \equiv \left( f \left( \pi \right) \right)_{H_j} \equiv \left( f \left( \pi\right) \right)_{u_j} := \left( u_j \right)_{\pi},
\]
where the coordinates of $\R^{m!}$ are indexed by elements of $S_m$. Then
\[
 \left( f \left( \sigma^n \right) \right)_j := \sum_{i=1}^n \left( f \left( \sigma_i \right) \right)_j = \sum_{i=1}^{n} \left( u_j \right)_{\sigma_i} = n \left( u_j \cdot x \left( \sigma^n \right) \right).
\]
The sign of $\left( f \left( \sigma^n \right) \right)_j$ thus tells us which side of the hyperplane $H_j$ the point $x \left( \sigma^n \right)$ lies on. We define $g \left( y \right)$ for all $y \in \R^{\ell+1}$ such that $y_{\ell + 1} = 0$; then the requirement that $g$ be compatible defines $g$ for all $y \in \R^{\ell +1}$. For $x\in \R$, define $\sgn \left( x \right)$ to be 1 if $x > 0$, $-1$ if $x < 0$, and 0 if $x = 0$. 

To define $g \left( y_1, \dots, y_{\ell}, 0 \right)$, look at the vector $\left( \sgn \left( y_1 \right), \dots, \sgn \left( y_{\ell} \right) \right)$. This vector determines a region in $\Delta^{m!}$ in the following way: if $\sgn \left( y_j \right) = 1$, then the region lies in the same open halfspace as $u_j$, if $\sgn \left( y_j \right) = - 1$ then the region lies in the open halfspace which does not contain $u_j$, and finally if $y_j = 0$, then the region lies in the hyperplane $H_j$. Now we define $g \left( y_1, \dots, y_{\ell}, 0 \right)$ to be the value of $F$ on the region of $\Delta^{m!}$ defined by $\left( \sgn \left( y_1 \right), \dots, \sgn \left( y_{\ell} \right) \right)$. The value of $g \left( y_1, \dots, y_{\ell}, 0 \right)$ is well-defined since $F$ is constant in each such region. Moreover, if we take $y \sim z$ with $y_{\ell + 1} = z_{\ell + 1} = 0$, then necessarily  $\left( \sgn \left( y_1 \right), \dots, \sgn \left( y_{\ell} \right) \right) =  \left( \sgn \left( z_1 \right), \dots, \
sgn \left( z_{\ell} \right) \right)$, and thus $g \left( y \right) = g \left( z \right)$: so $g$ is compatible (this is where we used the extra coordinate).

Now let us show that every generalized scoring rule is a hyperplane rule. Suppose a generalized scoring rule is given by functions $f : S_m \to \R^k$ and $g : \R^k \to \left[ m \right]$. For a ranking profile $\sigma^n \in S_m^n$, define $f \left( \sigma^n \right) := \sum_{i=1}^n f\left( \sigma_i \right) = n \sum_{\pi \in S_m} f \left( \pi \right) \left( x \left( \sigma^n \right) \right)_{\pi}$; in this way we can view $f$ as a function mapping $\N_{\geq 0}^{m!} \setminus \left\{ 0 \right\}$ to $\R^k$ (and hence can also view $\GS \left( f, g \right)$ as a function mapping $\N_{\geq 0}^{m!} \setminus \left\{ 0 \right\}$ to $\left[ m \right]$). Since this mapping  is homogeneous, we may extend the domain of $f$ (and hence that of $\GS \left( f, g \right)$) to $\Q_{\geq 0}^{m!} \setminus \left\{ 0 \right\}$ in the natural way. 

For a total preorder $\mcO$, let $R_{\mcO} = \left\{ x \in \Q_{\geq 0}^{m!} \setminus \left\{ 0 \right\} : f \left( x \right) \sim \mcO \right\}$. By definition, if $x,y \in R_{\mcO}$ then $g\left( f \left( x \right) \right) = g \left( f \left( y \right) \right)$, i.e., $\GS \left( f, g \right)$ is constant in each region $R_{\mcO}$. Each region $R_{\mcO}$ is a $\Q$-convex cone, i.e.\ if $x,y \in R_{\mcO}$ and $\lambda \in \Q \cap \left[ 0, 1 \right]$, then $\lambda x + \left( 1 - \lambda \right) y \in R_{\mcO}$, and furthermore if $\mu \in \Q_{>0}$, then $\mu x \in R_{\mcO}$ (both of these properties follow immediately from Definition~\ref{def:equiv}). Thus we can write $\Q_{\geq 0}^{m!} \setminus \left\{ 0 \right\}$ as the disjoint union of the $\Q$-convex cones $\left\{ R_{\mcO} \right\}_{\mcO}$. The only way to do this is by taking finitely many hyperplanes of $\R^{m!}$ and cutting $\Q_{\geq 0}^{m!} \setminus \left\{ 0 \right\}$ using these hyperplanes; a precise statement of this can be found in 
Appendix~\ref{sec:kemperman}. This essentially follows from a result by Kemperman~\cite[Theorem 2]{kemperman1986decomposing}---to keep the paper self-contained we reproduce in Appendix~\ref{sec:kemperman} his results and proof, and show how the statement above follows from his results. Since our function is homogeneous, we need only look at the values of $\GS \left( f, g \right)$ on the simplex $\Delta^{m!}$. By the above, the simplex is divided into regions $\left\{ R_{\mcO} \cap \Delta^{m!} \right\}_{\mcO}$ via affine hyperplanes of $\Delta^{m!}$, and the function $\GS \left( f, g \right)$ is constant on $R_{\mcO} \cap \Delta^{m!}$ for each total preorder $\mcO$, so $\GS \left( f, g \right)$ is indeed a hyperplane rule.
\end{proof}

\subsubsection{Examples}\label{sec:examples} 

Most commonly used SCFs are generalized scoring rules / hyperplane rules, including all positional scoring rules, instant-runoff voting, Coombs' method, contingent vote, the Kem\'eny-Young method, Bucklin voting, Nanson's method, Baldwin's method, Copeland's method, maximin, and ranked pairs. 
Some of these examples were already shown by Xia and Conitzer~\cite{xia2008generalized,xia2009finite}, but nevertheless in Appendix~\ref{sec:app-examples} we detail explanations of many of these examples. 
The main reason for this is that the perspective of a hyperplane rule arguably makes these explanations simpler and clearer. 
A rule that does not fit into this framework is Dodgson's rule, which is not homogeneous (see, e.g.,~\cite{brandt2009some}), and therefore it is not a hyperplane rule.

\subsection{Small coalitions for generalized scoring rules}\label{sec:small-subsec} 

We now show that for generalized scoring rules, a coalition of size $c \sqrt{n}$ for small enough $c$ can only change the outcome of the election with small probability. By the equivalence above, we can work in the framework of hyperplane rules.

We consider two metrics on $\Delta^{m!}$: the $L^1$ metric, denoted by $d_1$ or $\left\| \cdot \right\|_1$, and the $L^2$ metric, denoted by $d_2$ or $\left\| \cdot \right\|_2$. The $L^1$ metric is important in this setting, since changing the votes of voters corresponds to moving in the $L^1$ metric on $\Delta^{m!}$; this connection is formalized in the following lemma.

\begin{lemma}\label{lem:L1-Ham}
 Let $\sigma^n, \tau^{n} \in S_m^n$. Then $d_1 \left( x \left( \sigma^n \right), x \left( \tau^n \right) \right) \leq \frac{2}{n} d_{H} \left( \sigma^n, \tau^n \right)$, where $d_H$ denotes Hamming distance, i.e., $d_{H} \left( \sigma^n, \tau^n \right) = \sum_{i=1}^n \mathbf{1} \left[ \sigma_i \neq \tau_i \right]$. Furthermore, if $y \in D_n$, then there exists $\hat{\tau}^n \in S_m^n$ such that $x \left( \hat{\tau}^n \right) = y$ and $d_1 \left( x\left( \sigma^n \right), y \right) = \frac{2}{n} d_{H} \left( \sigma^n, \hat{\tau}^n \right)$.
\end{lemma}
\begin{proof}
Let $\pi^0 = \sigma^n$, and for $i = 1, \dots, n$, define the ranking profile $\pi^i$ as $\pi^i = \left( \tau_1, \dots, \tau_i, \sigma_{i+1}, \dots, \sigma_n \right)$. By definition, $\pi^n = \tau^n$. The desired inequality then follows from the triangle inequality:
\[
 d_1 \left( x \left( \sigma^n \right), x\left( \tau^n \right) \right) = d_1 \left( x \left( \pi^0 \right), x \left( \pi^n \right) \right) \leq \sum_{i=1}^n d_1 \left( x\left( \pi^{i-1} \right), x \left( \pi^i \right) \right) = \sum_{i=1}^n \frac{2}{n} \mathbf{1} \left[ \sigma_i \neq \tau_i \right] = \frac{2}{n} d_H \left( \sigma^n, \tau^n \right).
\]
For the second part of the lemma, construct $\hat{\tau}^n$ as follows. For each $\pi \in S_m$, let $I_{\pi} := \left\{ i \in \left[ n \right] : \sigma_i = \pi \right\}$. If $x \left( \sigma^n \right)_{\pi} \leq y_{\pi}$, then for every $i \in I_{\pi}$, let $\hat{\tau}_i = \pi$. If $x \left( \sigma^n \right)_{\pi} > y_{\pi}$, then choose a subset of indices $I'_{\pi} \subset I_{\pi}$ of size $\left| I'_{\pi} \right| = n y_{\pi}$, and for every $i \in I'_{\pi}$, let $\hat{\tau}_i = \pi$. Finally, define the rest of the coordinates of $\hat{\tau}^n$ so that $x \left( \hat{\tau}^n \right) = y$. The construction guarantees that then $d_1 \left( x \left( \sigma^n \right), y \right) = \frac{2}{n} d_H \left( \sigma^n, \hat{\tau}^n \right)$.
\end{proof}

It is therefore natural to define distances from the boundary $B$ using the $L^1$ metric:
\begin{definition}[Blowup of boundary]
 For $\alpha > 0$, we define the blowup of the boundary $B$ by $\alpha$ to be
\[
 B^{+\alpha} = \left\{ y \in \Delta^{m!} : \exists x \in B \text{ such that } \left\| x - y \right\|_1 \leq \alpha \right\}.
\]
\end{definition}
In order for some coalition to be able to change the outcome of the election at a given ranking profile, the point on the simplex corresponding to this ranking profile needs to be sufficiently close to the boundary~$B$; this is formulated in the following lemma.
\begin{lemma}\label{lem:necessary}
 Suppose we have $n$ voters, a coalition of size $k$, and the ranking profile is $\sigma^n \in S_m^n$, which corresponds to the point $x \left( \sigma^n \right) \in \Delta^{m!}$ on the probability simplex. A necessary condition for the coalition to be able to change the outcome of the election from this position is that $x \left( \sigma^n \right) \in B^{+2k/n}$. Conversely, if $x \left( \sigma^n \right) \in B^{+\left( 2k - m! \right) / n}$, then there exists a coalition of size $k$ that can change the outcome of the election.
\end{lemma}
\begin{proof}
 For any ranking profile $\tau^n$ that the coalition can reach, we have $d_H \left( \sigma^n, \tau^n \right) \leq k$, and so by Lemma~\ref{lem:L1-Ham} we have $d_1 \left( x \left( \sigma^n \right), x \left( \tau^n \right) \right) \leq \frac{2k}{n}$. If $x \left( \sigma^n \right) \notin B^{+2k/n}$, then for every ranking profile $\tau^n$ which the coalition can reach, $x \left( \sigma^{n} \right)$ and $x \left( \tau^n \right)$ are in the same region determined by the hyperplanes, and so $F \left( x  \left( \tau^n \right) \right) = F \left( x \left( \sigma^n \right) \right)$, i.e., the coalition cannot change the outcome of the election.

 Now suppose that $x \left( \sigma^n \right) \in B^{+\left( 2k - m! \right) / n}$. Then there exists $y \in B$ such that $d_1 \left( x\left( \sigma^n \right), y \right) \leq \frac{2k - m!}{n}$. Since $y \in B$, there exists $\hat{y} \in D_n$ such that $d_1 \left( y, \hat{y} \right) \leq \frac{m!}{n}$ and $F\left( \hat{y} \right) \neq F \left( x \left( \sigma^n \right) \right)$. By the triangle inequality, $d_1 \left( x \left( \sigma^n \right), \hat{y} \right) \leq \frac{2k}{n}$, and then by the second part of Lemma~\ref{lem:L1-Ham} there exists $\hat{\tau}^n \in S_m^n$ such that $x\left( \hat{\tau}^n \right) = \hat{y}$ and $d_H \left( \sigma^n, \hat{\tau}^n \right) \leq k$. The coalition consisting of voters with indices in $I := \left\{ i \in \left[ n \right] : \sigma_i \neq \hat{\tau}_i \right\}$ can thus change the outcome of the election.
\end{proof}
\begin{corollary}\label{cor:bdry_blowup}
 If we have $n$ voters, the probability that some coalition of size $k$ can change the outcome of the election is bounded from below by $\p \left( x\left( \sigma^n \right) \in B^{+\left( 2k - m! \right)/n} \right)$ and from above by $\p \left( x\left( \sigma^n \right) \in B^{+2k/n} \right)$, where $\sigma^n$ is drawn according to the probability distribution satisfying the conditions of the setup.
\end{corollary}
\noindent{\bf Gaussian limit.} Due to the i.i.d.-ness of the votes, the multinomial random variable $x \left( \sigma^n \right)$ concentrates around its expectation, and the rescaled random variable
\[
 \tilde{x} \left( \sigma^n \right) := \sqrt{n} \left( x \left( \sigma^n \right) - \E \left( x \left( \sigma^n \right) \right) \right)
\]
converges to a normal distribution, with zero mean and specific covariance structure. For our analysis it is better to use this Gaussian picture, and thus we will reformulate the preliminaries above in this limiting setting. First, let us determine the limiting distribution.

\begin{lemma} We have $\tilde{x} \left( \sigma^n \right) \Rightarrow_{n} N \left( 0, \Sigma \right)$, where the covariance structure is given by $\Sigma = \diag \left( p \right) - p p^T$, where recall that $p$ is the distribution of a vote.
\end{lemma}
\begin{proof}
It is clear that $\E \left( \tilde{x} \left( \sigma^n \right) \right) = 0$. Computing the covariance structure, we first have that $\E \left( x_{\pi}^2 \right) = \frac{1}{n^2} \sum_{i,j=1}^n \p \left( \sigma_i = \pi , \sigma_j = \pi \right) = \left( 1 - \frac{1}{n} \right) p \left( \pi \right)^2 + \frac{1}{n} p \left( \pi \right)$, 
from which we have $\Var \left( x_{\pi} \right) = \frac{1}{n} \left( p \left( \pi \right) - p \left( \pi \right)^2 \right)$ and thus $\Var \left( \tilde{x}_{\pi} \right) = p \left( \pi \right) - p \left( \pi \right)^2$. Then similarly for $\pi \neq \pi'$ we have $\E \left( x_{\pi} x_{\pi'} \right) = \frac{1}{n^2} \sum_{i,j=1}^n \p \left( \sigma_i = \pi, \sigma_j = \pi' \right) = \frac{1}{n^2} \sum_{i\neq j} p \left( \pi \right) p \left( \pi' \right) = \left( 1 - \frac{1}{n} \right) p \left( \pi \right) p \left( \pi' \right)$, 
from which we have that $\Cov \left( x_{\pi}, x_{\pi'} \right) = - \frac{1}{n} p \left( \pi \right) p \left( \pi' \right)$ and thus $\Cov \left( \tilde{x}_{\pi}, \tilde{x}_{\pi'} \right) = - p \left( \pi \right) p \left( \pi' \right)$.
\end{proof}

Note: because of the concentration of $x\left( \sigma^n \right)$ around its mean, and our assumption that for every $n$ and for every candidate $a \in \left[ m \right]$, $\p \left( f \left( \sigma^n \right) = a \right) \geq \eps$, it is necessary that for every $\alpha > 0$ and for every candidate $a \in \left[ m \right]$ there exists $y \in \Delta^{m!}$ such that $\left\| y - \E \left( x \left( \sigma_1 \right) \right) \right\|_1 \leq \alpha$ and $F \left( y \right) = a$.

Denote by $\mu$ the distribution of $N \left( 0, \Sigma \right)$ and let $\tilde{X}$ denote a random variable distributed according to $\mu$. Note that $\mu$ is a degenerate multivariate normal distribution, as the support of $\mu$ concentrates on the hyperplane $H_0$ where the coordinates sum to zero. (This is because $\sum_{\pi \in S_m} \tilde{x} \left( \sigma^n \right)_{\pi} = 0$.)

The underlying function $F : \Delta^{m!} \to \left[ m \right]$ corresponds to a function $\tilde{F} : \R^{m!}|_{H_0} \to \left[ m \right]$ in the Gaussian limit, and this function $\tilde{F}$ partitions $\R^{m!}|_{H_0}$ into $m$ parts based on the outcome of $\tilde{F}$. We denote these parts by $\left\{ \tilde{W}_a \right\}_{a\in \left[m\right]}$. We will need the following definitions and properties of boundaries, analogous to those above.

\begin{definition}[Interior and boundaries of a partition]
 We say that $\tilde{x} \in \R^{m!}|_{H_0}$ is an interior point of the partition $\left\{\tilde{W}_a \right\}_{a \in \left[m\right]}$ induced by $\tilde{F}$ if there exists $\alpha > 0$ such that for all $\tilde{y} \in \R^{m!}|_{H_0}$ for which $\left\| \tilde{x} - \tilde{y} \right\|_1 \leq \alpha$, we have $\tilde{F}\left( \tilde{x} \right) = \tilde{F} \left( \tilde{y} \right)$. Otherwise, we say that $\tilde{x} \in \R^{m!}|_{H_0}$ is on the boundary of the partition, which we denote by $\tilde{B}$.
\end{definition}

\begin{lemma}\label{lem:bdry_limit}
If the boundary $B$ comes from a hyperplane rule, i.e., $B$ is contained in the union of $\ell$ affine hyperplanes in $\Delta^{m!}$, then $\tilde{B}$ is contained in the union of $\tilde{\ell}$ hyperplanes of $\R^{m!}|_{H_0}$, where $\tilde{\ell} \leq \ell$.
\end{lemma}
\begin{proof}
 Two things can happen to an affine hyperplane $H$ of $\Delta^{m!}$ when we take the Gaussian limit: (1) if $\E \left( x \left( \pi \right) \right) \in H$, then translation by $\E \left( x \left( \pi \right) \right)$ takes $H$ into a hyperplane $\tilde{H}$ of $\R^{m!}|_{H_0}$, and since $\tilde{H}$ goes through the origin, scaling (in particular by $\sqrt{n}$) does not move this hyperplane; (2) if $\E \left( x \left( \pi \right) \right) \notin H$, then translation by $\E \left( x \left( \pi \right) \right)$ takes $H$ into an affine hyperplane $\tilde{H}$ of $\R^{m!}|_{H_0}$ that does not go through the origin, and then scaling by $\sqrt{n}$ moves $\tilde{H}$ to an affine hyperplane of $\R^{m!}|_{H_0}$ whose $L^2$ distance from the origin is proportional to $\sqrt{n}$, so in the $n \to \infty$ limit this affine hyperplane ``vanishes''.
\end{proof}

\begin{definition}[Blowup of boundary]
 For $\alpha > 0$, we define the blowup of the boundary $\tilde{B}$ by $\alpha$ to be
\[
 \tilde{B}^{+\alpha} = \left\{ \tilde{y} \in \R^{m!}|_{H_0} : \exists \tilde{x} \in \tilde{B} \text{ such that } \left\| \tilde{x} - \tilde{y} \right\|_1 \leq \alpha \right\}.
\]
\end{definition}

Let us focus specifically on a coalition of size $c \sqrt{n}$ for some (small) constant $c$. Corollary~\ref{cor:bdry_blowup} implies the following.

\begin{corollary}\label{cor:limit}
 For hyperplane rules the limit of the probability that in an election with $n$ voters some coalition of size $c \sqrt{n}$ can change the outcome of the election is $\mu \left( \tilde{X} \in \tilde{B}^{+2c} \right)$.
\end{corollary}

The following claim, together with Corollary \ref{cor:limit}, tells us that for hyperplane rules a coalition of size $c \sqrt{n}$ can change the outcome of the election with only small probability, given that $c$ is sufficiently small, proving Part~\ref{small} of Theorem~\ref{thm:main2}.

\begin{claim}\label{claim:hyperplanes_limit}
 Suppose our SCF is a hyperplane rule, and in particular let $\left\{ \tilde{H}_i \right\}_{i=1}^{M}$ be a collection of hyperplanes in $\R^{m!}|_{H_{0}}$ such that $\tilde{B} \subseteq \bigcup_{i=1}^{M} \tilde{H}_i$. Then
\[
\mu \left( \tilde{X} \in \tilde{B}^{+c} \right) \leq \sqrt{\frac{2}{\pi}} \frac{Mc}{\sqrt{\delta}}.
\]
\end{claim}

\begin{proof}
 By our condition and a union bound we have
\[
 \mu \left( \tilde{X} \in \tilde{B}^{+c} \right) \leq \sum_{i=1}^{M} \mu \left( \tilde{X} \in \tilde{H}_i^{+c} \right).
\]
For a hyperplane $\tilde{H}$ in $\R^{m!}|_{H_0}$, denote (one of) the corresponding unit normal vector(s) (in the hyperplane $H_0$) by $u$. Then 
\[
\tilde{H} = \left\{ \tilde{x} \in \R^{m!}|_{H_0} : u \cdot \tilde{x} = 0 \right\}
\]
and since $L^1$ distance is always greater than $L^2$ distance, we have
\[
 \tilde{H}^{+c} \subseteq \left\{ \tilde{x} \in \R^{m!}|_{H_0} : \exists \tilde{y} \in \tilde{H} \text{ such that } \left\| \tilde{x} - \tilde{y} \right\|_2 \leq c \right\} = \left\{ \tilde{x} \in \R^{m!}|_{H_0} : \left| u \cdot \tilde{x} \right| \leq c \right\}.
\]
Since $\tilde{X}$ is a multidimensional Gaussian r.v., $u \cdot \tilde{X}$ is a one-dimensional Gaussian r.v.\ (which is centered). Therefore
\[
 \mu \left( \tilde{X} \in \tilde{H}^{+c} \right) \leq \mu \left( u \cdot \tilde{X} \in \left[ -c, c \right] \right) \leq \frac{2c}{\sqrt{2\pi \Var \left( u \cdot \tilde{X} \right)}}.
\]
We have that
\[
 \Var \left( u \cdot \tilde{X} \right) = \E \left( u \cdot \tilde{X} \right)^2 = \E \left( u^T \tilde{X} \tilde{X}^T u \right) = u^T \Sigma u,
\]
and so all that remains to show is that
\[
 \min_{u : \left\| u \right\| = 1, u \perp \mathbf{1}} u^T \Sigma u \geq \delta,
\]
where $\mathbf{1}$ is the $m!$-dimensional vector having 1 in every coordinate.

Let $\lambda_1 \left( \Sigma \right) \geq \lambda_2 \left( \Sigma \right) \geq \dots \geq \lambda_{m!} \left( \Sigma \right)$ denote the eigenvalues of $\Sigma$. Since $\Sigma$ is positive semidefinite, all eigenvalues are nonnegative. We know that 0 is an eigenvalue of $\Sigma$ (the corresponding eigenvector is $\mathbf{1}$), so $\lambda_{m!} \left( \Sigma \right) = 0$. By the variational characterization of eigenvalues we have
\[
 \min_{u : \left\| u \right\| = 1, u \perp \mathbf{1}} u^T \Sigma u = \lambda_{m!-1} \left( \Sigma \right),
\]
and so we need to show that $\lambda_{m!-1} \left( \Sigma \right) \geq \delta$. To do this we use Weyl's inequalities.
\begin{lemma}[Weyl's inequalities]
 For an $n \times n$ matrix $M$ let $\lambda_1 \left( M \right) \geq \lambda_2 \left( M \right) \geq \dots \geq \lambda_n \left( M \right)$ denote its eigenvalues. If $A$ and $C$ are $n \times n$ symmetric matrices then
  \begin{align*}
    \lambda_j \left( A + C \right) &\leq \lambda_i \left( A \right) + \lambda_{j-i+1} \left( C \right) \qquad \text{ if } i \leq j,\\
    \lambda_j \left( A + C \right) &\geq \lambda_i \left( A \right) + \lambda_{j-i+n} \left( C \right) \qquad \text{ if } i \geq j.
  \end{align*} 
\end{lemma}
We use Weyl's inequality for $A = \diag \left( p \right)$ and $C = - pp^T$. The eigenvalues of $A$ are $\left\{ p \left( \pi \right) \right\}_{\pi \in S_m}$, all of which are no less than $\delta$. Since $C$ has rank 1, all its eigenvalues but one are zero, and the single nonzero eigenvalue is $\lambda_{m!} \left( C \right) = - p^T p$. Since $\Sigma = \diag \left( p \right) - p p^T = A + C$, Weyl's inequality tells us that
\[
 \lambda_{m!-1} \left( \Sigma \right) \geq \lambda_{m!} \left( \diag \left( p \right) \right) + \lambda_{m!-1} \left( - p p^T \right) \geq \delta + 0 = \delta. \qedhere
\]
\end{proof}

This implies that we have a lower bound of $\Omega \left( \sqrt{n} \right)$ for the size of the coalition needed in order to change the outcome of the election for hyperplane rules. As mentioned before, most commonly occurring SCFs are in this class of rules: see Appendix~\ref{sec:app-examples} for many examples.

\subsubsection{``Almost'' hyperplane rules}\label{sec:almost} 

Furthermore, the Gaussian limiting setting above is not sensitive to small changes to the voting rule for finite $n$. Consequently, for SCFs that are ``almost'' hyperplane rules (in a sense we make precise below), the same conclusion holds: a coalition of size $\Omega \left( \sqrt{n} \right)$ is needed in order to be able to change the outcome of the election with non-negligible probability. In particular, the same result holds for SCFs with arbitrary tie-breaking rules for ranking profiles which lie on one of the hyperplanes (e.g., the tie-breaking rule can depend on the number of voters $n$). 

\begin{definition}[``Almost'' hyperplane rules]
 Fix a finite set of affine hyperplanes of the simplex $\Delta^{m!}$: $H_1, \dots, H_\ell$. These partition the simplex into finitely many regions. Let $F : \Delta^{m!} \to \left[m \right]$ be a function which is constant on each such region, and let $B$ denote the induced boundary. Then the sequence of SCFs $\left\{ f_n \right\}_{n \geq 1}$, $f_n : S_m^n \to \left[ m \right]$, is called an ``almost'' hyperplane rule if for every $\sigma^n$ such that $x \left( \sigma^n \right) \notin B^{+o \left( 1 / \sqrt{n} \right)}$, we have
\[
 f_n \left( \sigma^n \right) = F \left( x \left( \sigma^n \right) \right).
\]
This SCF is called an ``almost'' hyperplane rule induced by the affine hyperplanes $H_1, \dots, H_{\ell}$ and the function~$F$.
\end{definition}

\begin{lemma}
Suppose the sequence of SCFs $\left\{ f_n \right\}_{n \geq 1}$, $f_n : S_m^n \to \left[ m \right]$, is an ``almost'' hyperplane rule defined by $\ell$ hyperplanes. Then in the Gaussian limiting setting the boundary $\tilde{B}$ is contained in the union of $\tilde{\ell}$ hyperplanes of $\R^{m!}|_{H_0}$, where $\tilde{\ell} \leq \ell$.
\end{lemma}
\begin{proof}
 For finite $n$, the induced boundary of $f_n$ in the simplex $\Delta^{m!}$ is contained in $B^{+ o \left( 1 / \sqrt{n} \right)}$, by definition. Since in the Gaussian limit we scale by $\sqrt{n}$, the blowup by $o \left( 1 / \sqrt{n} \right)$ of the boundary $B$ disappears in the limit, and hence we are back to the situation of Lemma \ref{lem:bdry_limit}. Consequently, the affine hyperplanes corresponding to our ``almost'' hyperplane rule either ``disappear to infinity'' or become hyperplanes of $\R^{m!}|_{H_0}$. 
\end{proof}

\begin{corollary}
 Corollary \ref{cor:limit} and Claim \ref{claim:hyperplanes_limit} hold for ``almost'' hyperplane rules as well.
\end{corollary}

\subsection{Smoothness of the phase transition}\label{sec:smooth} 

In this final subsection our goal is to show Parts~\ref{limit} and~\ref{smooth} of Theorem~\ref{thm:main2}. The existence of the limits in Part~\ref{limit} follows immediately from the Gaussian limit described above; we do not detail this, but rather give formulas for these limiting probabilities. These then imply the properties described in Part~\ref{smooth} of the theorem.

In the following let the hyperplane rule be given by affine hyperplanes $H_1, \dots, H_{\ell}$ of $\Delta^{m!}$ and the function $F : \Delta^{m!} \to \left[ m \right]$; in the limiting setting denote by $\tilde{H}_1, \dots, \tilde{H}_{\tilde{\ell}}$ the corresponding hyperplanes of $\R^{m!}|_{H_0}$ and denote by $\tilde{F} : \R^{m!}|_{H_0} \to \left[ m \right]$ the corresponding function.

\subsubsection{The quantities $\overline{q}$ and $\underline{q}$}\label{sec:smooth_some} 

For $\tilde{x} \in \R^{m!}|_{H_0}$ define
\[
 \alpha \left( \tilde{x} \right) := \inf_{\tilde{y} : \tilde{F} \left( \tilde{y} \right) \neq \tilde{F} \left( \tilde{x} \right)}  d_1 \left( \tilde{x}, \tilde{y} \right), \qquad \qquad \qquad \beta \left( \tilde{x} \right) := \max_{a \in \left[ m \right]} \inf_{\tilde{y} : \tilde{F} \left( \tilde{y} \right) = a} d_1 \left( \tilde{x}, \tilde{y} \right).
\]
From the previous subsection it is then immediate that we can write
\begin{align*}
 \overline{q} \left( c \right) &= \mu \left( \tilde{X} : \alpha \left( \tilde{X} \right) \leq 2c \right),\\
 \underline{q} \left( c \right) &= \mu \left( \tilde{X} : \beta \left( \tilde{X} \right) \leq 2c \right).
\end{align*}
It is important to note that the boundary $\tilde{B}$ is contained in the union of finitely many hyperplanes, $\tilde{H}_1, \dots, \tilde{H}_{\tilde{\ell}}$, and thus the regions where $\tilde{F}$ is constant are convex cones which are the intersection of finitely many halfspaces. Consequently $\alpha \left( \tilde{x} \right)$ is either $d_1 \left( \tilde{x}, 0 \right)$, where $0$ denotes the origin of $\R^{m!}$, or it is $d_1 \left( \tilde{x}, \tilde{H}_j \right)$ for some $1 \leq j \leq \tilde{\ell}$, where $d_1 \left( \tilde{x}, \tilde{H}_j \right) = \inf_{\tilde{y} \in \tilde{H}_j} d_1 \left( \tilde{x}, \tilde{y} \right)$. 
If we scale $\tilde{x}$ by some positive constant $\lambda$, then the distance from the origin and from every hyperplane scales as well (i.e., $d_1 \left( \lambda \tilde{x}, 0 \right) = \lambda d_1 \left( \tilde{x}, 0 \right)$ and $d_1 \left( \lambda \tilde{x}, \tilde{H}_j \right) = \lambda d_1 \left( \tilde{x}, \tilde{H}_j \right)$), and thus for every $\lambda > 0$, we have $\alpha \left( \lambda \tilde{x} \right) = \lambda \alpha \left( \tilde{x} \right)$. 
Consequently, if we write $\tilde{x} = \left\| \tilde{x} \right\|_2 \tilde{s}$, where $\tilde{s} \in S^{m! - 1}$, and $S^{m!-1}$ denotes the $\left( m! - 1 \right)$-sphere (not to be confused with $S_m^n$, the set of ranking profiles on $n$ voters and $m$ candidates), then we have $\alpha \left( \tilde{x} \right) = \left\| \tilde{x} \right\|_2 \alpha \left( \tilde{s} \right)$. 

The same scaling property holds for $\beta$ as well, and hence we have
\begin{align}
 \overline{q} \left( c \right) &= \mu \left( \tilde{X} : \left\| \tilde{X} \right\|_2 \alpha \left( \tilde{S} \right) \leq 2c \right), \label{eq:q1}\\
 \underline{q} \left( c \right) &= \mu \left( \tilde{X} : \left\| \tilde{X} \right\|_2 \beta \left( \tilde{S} \right) \leq 2c \right).\label{eq:q2}
\end{align}

Recall that our condition that for every $a \in \left[ m \right]$, $\p \left( f \left( \sigma \right) = a \right) \geq \eps$, implies that for every $\eta > 0$ and for every $a \in \left[ m \right]$ there exists $\tilde{x} \in \R^{m!}|_{H_0}$ such that $\left\| \tilde{x} \right\|_2 \leq \eta$ and $\tilde{F} \left( \tilde{x} \right) = a$. Consequently for every $\tilde{x} \in \R^{m!}|_{H_0}$ we must have $\alpha \left( \tilde{x} \right) \leq d_1 \left( \tilde{x}, 0 \right)$ and $\beta \left( \tilde{x} \right) \leq d_1 \left( \tilde{x}, 0 \right)$. In particular, for $\tilde{s} \in S^{m!-1}$ we have $d_1 \left( \tilde{s}, 0 \right) \leq \sqrt{m!} d_2 \left( \tilde{s}, 0 \right) = \sqrt{m!}$ and so $\alpha \left( \tilde{s} \right), \beta \left( \tilde{s} \right) \leq \sqrt{m!}$. This immediately implies that for every $c > 0$ we have
\[
 \underline{q} \left( c \right) \geq \mu \left( \tilde{X} : \left\| \tilde{X} \right\|_2 \leq \frac{2c}{\sqrt{m!}} \right) > 0.
\]

To show that $\overline{q} \left( c \right) < 1$, note that since the boundary is contained in the union of finitely many hyperplanes, there exists $\tilde{s}^{*} \in S^{m! - 1}$ such that $\alpha \left( \tilde{s}^{*} \right) > 0$. By continuity of $\alpha$, there exists a neighborhood $U \subseteq S^{m!-1}$ of $\tilde{s}^{*}$ such that for every $\tilde{s} \in U$, $\alpha \left( \tilde{s} \right) \geq \alpha \left( \tilde{s}^{*} \right) / 2$. For any $\tilde{x}$ such that $\tilde{x} / \left\| \tilde{x} \right\|_2 \in U$ and $\left\| \tilde{x} \right\|_2 > \frac{4c}{\alpha \left( \tilde{s}^{*} \right)}$, we have
\[
 \alpha \left( \tilde{x} \right) = \left\| \tilde{x} \right\|_2 \alpha \left( \tilde{x} / \left\| \tilde{x} \right\|_2 \right) > \frac{4c}{\alpha \left( \tilde{s}^{*} \right)} \frac{\alpha \left( \tilde{s}^{*} \right)}{2} = 2c.
\]
So consequently
\[
 \overline{q} \left( c \right) \leq 1 - \mu \left( \tilde{X} : \tilde{X} / \left\| \tilde{X} \right\|_2 \in U , \left\| \tilde{X} \right\|_2 > \frac{4c}{\alpha \left( \tilde{s}^{*} \right)}   \right) < 1.
\]

Finally, the fact that $\underline{q} \left( c \right)$ and $\overline{q} \left( c \right)$ are continuously differentiable follows from the formulas \eqref{eq:q1} and \eqref{eq:q2}, since $\underline{q} \left( c \right)$ and $\overline{q} \left( c \right)$ are both written as the Gaussian volume of a subset of $\R^{m!}|_{H_0}$, and in both cases this subset grows continuously as $c$ increases. The derivative of both $\underline{q} \left( c \right)$ and $\overline{q} \left( c \right)$ is bounded at zero (by Corollary~\ref{cor:limit} and Claim~\ref{claim:hyperplanes_limit}), while as $c \to \infty$ the derivative approaches zero, and since the derivative is continuous, it must be bounded by a constant for the whole half-line.  

\subsubsection{The quantities $\overline{r}$ and $\underline{r}$}\label{sec:smooth_specific} 

In the previous setup when the coalition of size $c \sqrt{n}$ was not specified, the ranking profile could be changed arbitrarily within a Hamming ball of radius $c \sqrt{n}$. On the probability simplex $\Delta^{m!}$ this corresponded to an $L^1$ ball of radius $2c / \sqrt{n}$, and in the rescaled limiting setting it corresponded to an $L^1$ ball in $\R^{m!}|_{H_0}$ of radius $2c$. When the coalition of size $c \sqrt{n}$ is specified, things are slightly different. In particular, when we look at the probability distribution on the probability simplex $\Delta^{m!}$ induced by the distribution on ranking profiles (or, in the limiting setting, the Gaussian distribution on $\R^{m!}|_{H_0}$), then we have lost track of the votes of any specific coalition. Nonetheless, the Gaussian limiting setting still provides formulas for the limiting probabilities $\underline{r} \left( c \right)$ and $\overline{r} \left( c \right)$.

We can first draw a random ranking profile for the other $n - c \sqrt{n}$ voters not in the coalition, $\sigma^{n - c\sqrt{n}}$, and then the voters in the coalition can set their votes arbitrarily. The question is, how can the coalition affect the outcome of the vote? In particular, (a) can they change the outcome of the election, and (b) can they elect any candidate? 

The ranking profile $\sigma^{n - c\sqrt{n}}$ corresponds to a point $x \left( \sigma^{n - c\sqrt{n}} \right)$ on the probability simplex $\Delta^{m!}$, and by setting their votes the coalition can move this point on the probability simplex in some neighborhood of $x \left( \sigma^{n - c\sqrt{n}} \right)$. We omit the calculation for finite $n$ and only present the result in the limiting setting. 

Suppose the limiting ranking profile of the voters other than the coalition corresponds to the point $\tilde{x} \in \R^{m!}|_{H_0}$. Then the set of points the coalition can reach is the following:
\[
 R_c \left( \tilde{x} \right) := \left\{ \tilde{y} \in \R^{m!}|_{H_0} : \forall \pi \in S_m : \tilde{y}_{\pi} - \tilde{x}_{\pi} + c p \left( \pi \right) \geq 0 \right\}.
\]
We can then define
\begin{align*}
 \varphi \left( \tilde{x} \right) &:= \inf \left\{ \gamma : \exists \tilde{y} \in R_{\gamma} \left( \tilde{x} \right) \text{ such that } \tilde{F} \left( \tilde{y} \right) \neq \tilde{F} \left( \tilde{x} \right) \right\},\\
 \psi \left( \tilde{x} \right) &:= \inf \left\{ \gamma : \forall a \in \left[m \right] \exists \tilde{y} \in R_{\gamma} \left( \tilde{x} \right) \text{ such that } \tilde{F} \left( \tilde{y} \right) = a \right\},
\end{align*}
and it follows immediately that we can then write
\begin{align*}
 \overline{r} \left( c \right) &= \mu \left( \tilde{X} : \varphi \left( \tilde{X} \right) \leq c \right),\\
 \underline{r} \left( c \right) &= \mu \left( \tilde{X} : \psi \left( \tilde{X} \right) \leq c \right).
\end{align*}
In the same way as in Section~\ref{sec:smooth_some} one can argue that $\varphi$ and $\psi$ scale: if $\lambda > 0$ then $\varphi \left( \lambda \tilde{x} \right) = \lambda \varphi \left( \tilde{x} \right)$ and $\psi \left( \lambda \tilde{x} \right) = \lambda \psi \left( \tilde{x} \right)$. Hence we have
\begin{align}
 \overline{r} \left( c \right) &= \mu \left( \tilde{X} : \left\| \tilde{X} \right\|_2 \varphi \left( \tilde{S} \right) \leq c \right), \label{eq:r1}\\
 \underline{r} \left( c \right) &= \mu \left( \tilde{X} : \left\| \tilde{X} \right\|_2 \psi \left( \tilde{S} \right) \leq c \right).\label{eq:r2}
\end{align}
For every $0 < c < \infty$ we have $\overline{r} \left( c \right) \leq \overline{q} \left( c \right) < 1$ (using Section~\ref{sec:smooth_some}). Let us now show that also $\underline{r} \left( c \right) > 0$. We claim that for all $\tilde{s} \in S^{m!-1}|_{H_0}$, $\psi \left( \tilde{s} \right) \leq \frac{2}{\delta}$. This follows from the fact that if $\tilde{s} \in S^{m!-1}|_{H_0}$ then $S^{m!-1}|_{H_0} \subseteq R_{\frac{2}{\delta}} \left( \tilde{s} \right)$, which is true because if $\tilde{y} \in S^{m!-1}|_{H_0}$ then for all $\pi \in S_m$, $\tilde{y}_{\pi} - \tilde{s}_{\pi} + \frac{2}{\delta} p \left( \pi \right) \geq -1 -1 + \frac{2}{\delta} \delta = 0$.
Thus we have
\[
 \underline{r} \left( c \right) \geq \mu \left( \tilde{X} : \left\| \tilde{X} \right\|_2 \leq \frac{c\delta}{2} \right) > 0
\]
as claimed.

Finally, the fact that $\underline{r} \left( c \right)$ and $\overline{r} \left( c \right)$ are continuously differentiable follows from the formulas \eqref{eq:r1} and \eqref{eq:r2} using an argument given above: $\underline{r} \left( c \right)$ and $\overline{r} \left( c \right)$ are written as the Gaussian volume of subsets of $\R^{m!}|_{H_0}$, and these subsets grow continuously as $c$ increases. The derivative of both $\underline{r} \left( c \right)$ and $\overline{r} \left( c \right)$ is bounded at zero (by Corollary~\ref{cor:limit} and Claim~\ref{claim:hyperplanes_limit}), while as $c \to \infty$ the derivative approaches zero, and since the derivative is continuous, it must be bounded by a constant for the whole half-line. \qed


\section*{Acknowledgments}

We thank anonymous referees for helpful comments.


\bibliographystyle{plain}
\bibliography{coal_manip,abb,ultimate}


\appendix

\section{\texorpdfstring{Decomposing $\R^d$ as the disjoint union of finitely many convex cones: only via hyperplanes}{Decomposing R^{d} as the disjoint union of finitely many convex cones: only via hyperplanes}}\label{sec:kemperman} 

For self-containment, we reproduce here the main definitions and results of Kemperman~\cite{kemperman1986decomposing} that make precise the claim used in the proof of Lemma~\ref{lem:equiv} that the only way to decompose $\Q_{\geq 0}^d \setminus \left\{ 0 \right\}$ into the disjoint union of finitely many $\Q$-convex cones is via hyperplanes. Kemperman's paper deals with convex sets in general, but here we summarize the results about convex \emph{cones} that are relevant to us. Kemperman's results pertain to finite dimensional linear spaces and we will state them in this form; in the end we show how results for $\R_{\geq 0}^d$ follow immediately from these, and as a consequence we also obtain the claim used in the proof of Lemma~\ref{lem:equiv}.

Let us start with the main definitions. In the following, all linear spaces are over the reals and are finite dimensional. Let $X$ be a linear space. A convex cone is a subset $K \subseteq X$ such that $x,y \in K$ and $\lambda > 0$ imply $x + y \in K$ and $\lambda x \in K$. (We do not require that $0 \in K$.) For a set $A \subseteq X$, denote its affine hull by $\aff \left( A \right)$, its convex hull by $\cvx \left( A \right)$, and its closure by $\cl \left( A \right)$. Note that if $K \subseteq X$ is a convex cone, then $\aff \left( K \right)$ is a linear subspace of $X$.

We define two special types of convex cones: basic convex cones and elementary convex cones.

\begin{definition}[Basic convex cone]
 Let $K$ be a convex cone in a finite dimensional linear space $X$. We say that $K$ is a \emph{basic convex cone} (in $X$) if $K$ is a member $K = K_0$ of some partition
\[
 X =  K_0 \dot{\cup} K_1 \dot{\cup} \dots \dot{\cup} K_r
\]
 of $X$ into finitely many disjoint convex cones $\left\{K_i \right\}_{i=0}^r$.
\end{definition}
Note that any linear subspace $Y$ of $X$ is a basic convex cone, from which it immediately follows that $K$ is a basic convex cone in $X$ if and only if it is a basic convex cone in $\aff \left( K \right)$. 

In order to define elementary convex cones, we need a few more definitions.

\begin{definition}[Open polyhedral convex cone]
 Let $K$ be a convex cone in a finite dimensional linear space $X$. We say that $K$ is an \emph{open polyhedral convex cone} relative to $X$ if $K$ can be expressed as the intersection of finitely many open halfspaces $H_1, \dots, H_{\ell}$ of $X$, each of which has the origin on its boundary. The whole linear space $X$ is an open polyhedral convex cone with $\ell=0$.
\end{definition}

\begin{definition}[Relatively open polyhedral convex cone]
 Let $K$ be a convex cone in a finite dimensional linear space $X$. Then $K$ is a \emph{relatively open polyhedral convex cone} if either $K = \emptyset$ or $K$ is an open polyhedral convex cone relative to $\aff \left( K \right)$.
\end{definition}

\begin{definition}[Elementary convex cone]
 Let $K$ be a convex cone in a finite dimensional linear space $X$. We say that $K$ is an \emph{elementary convex cone} if $K$ can be represented as a disjoint union of finitely many relatively open polyhedral convex cones.
\end{definition}

The main result of Kemperman concerning convex cones is the following~\cite[Theorem 2]{kemperman1986decomposing}.

\begin{theorem}\label{thm:kemperman}
 Let $K$ be a convex cone in $\R^d$. Then $K$ is a basic convex cone if and only if it is an elementary convex cone.
\end{theorem}

In Lemma~\ref{lem:equiv} we only use the ``only if'' direction, and we thus leave the proof of the ``if'' direction as an exercise for the reader.

\begin{proof}[Proof of ``only if'' direction.]
 Let $X$ be a finite dimensional linear space and let $K$ be a basic convex cone in $X$ of dimension $d = \dim \left( K \right) =  \dim \left( Y \right)$, where $Y = \aff \left( K \right)$. We prove by induction on $d$ the following:

\begin{enumerate}[(i)]
 \item The relative interior of $K$, denoted by $K^0$, is a relatively open polyhedral convex cone.
 \item If $K^0 \neq Y$, then denote by $F_1, \dots, F_{\ell}$ the $\left( d - 1 \right)$-dimensional hyperplanes in $Y$ corresponding to the finitely many faces of the polyhedron $\cl \left( K \right) = \cl \left( K^0 \right)$. Then the convex cones $F_i \cap K$, $i = 1, \dots, \ell$, are elementary convex cones of dimension at most $d - 1$ (but they need not be disjoint).
 \item The convex cone $K$ is also an elementary convex cone.
\end{enumerate}

 If $K = \emptyset$, then properties (i) - (iii) hold. If $d = 0$, then necessarily $K = \left\{ 0 \right\}$, since $K$ is a convex cone, and again $K$ satisfies properties (i) - (iii) above.

 So we may assume that $d \geq 1$ and that each basic convex cone of dimension at most $d-1$ satisfies properties (i) - (iii) above. Since $K$ is a basic convex cone, there exists a partition
\begin{equation}\label{eq:partition}
 Y = K_0 \dot{\cup} K_1 \dot{\cup} \dots \dot{\cup} K_r
\end{equation}
of $Y$ into finitely many disjoint convex cones $\left\{ K_j \right\}_{j=0}^r$, with $K_0 = K$. We may assume that $r \geq 0$ is minimal, and hence the $K_j$ are non-empty. Note that $K^0$ is also non-empty since $\dim \left( K \right) = \dim \left( Y \right)$.

If $r=0$ then $K= K_0 = Y$ and the properties (i) - (iii) above are immediately satisfied, so we may assume that $r \geq 1$. For $j = 1, \dots, r$, let $H_j$ be a hyperplane in $Y$ which separates the convex cone $K = K_0$ with non-empty interior $K^0$ from the non-empty convex cone $K_j$. (Such hyperplanes exist by the hyperplane separation theorem, and, moreover, each such hyperplane goes through the origin, because each $K_j$ contains at least one point from every open ball around the origin, since each $K_j$ is a cone.) Let $H_j^0$ be the associated open half space in $Y$ which contains the interior $K^0$ of $K$. Let
\[
 L^0 = H_1^0 \cap \dots \cap H_r^0.
\]
Then $L^0$ is a polyhedral convex cone, which is open relative to $Y$, and contains the interior $K^0$ of $K$.

We claim that $L^0 = K^0$. It is enough to show that $L^0 \subseteq K$, because then $L^0 \subseteq K^0$ follows from the definition of $K^0$. Suppose on the contrary that there exists $x \in L^0$ such that $x \notin K$. Then from the partition \eqref{eq:partition} there must exist an index $1 \leq j \leq r$ with $x \in K_j$. This implies that $x \notin H_j^0$ and thus $x \notin L^0$, which is a contradiction. This proves (i).

Now let us show (ii). By \eqref{eq:partition}, we can write the linear space $F_i$ as the disjoint union of the convex cones $F_i \cap K_j$, $j = 0, \dots, r$, and thus $F_i \cap K$ is a basic convex cone and hence, by induction, an elementary convex cone.

Finally, let us show that $K$ is an elementary convex cone. Since $K^0$ is a polyhedral convex cone which is open relative to $Y$, it only remains to show that $K \setminus K^0$ can be written as a finite disjoint union of relatively open polyhedral convex cones. By (ii), we can write $K \setminus K^0$ as the finite union of elementary convex cones:
\[
 K \setminus K^0 = \cup_{i=1}^{\ell} \left( F_i \cap K \right),
\]
so what remains is to show that we can write this as a finite \emph{disjoint} union of relatively open polyhedral convex cones. We may assume w.l.o.g.\ that $F_i \cap K \neq \emptyset$ for all $i$ and that $\left( F_i \cap K \right) \nsubseteq \left( F_j \cap K \right)$ for all $i \neq j$ (otherwise we can leave out $F_i \cap K$ from the union).

We claim that then for every $i$,
\begin{equation}\label{eq:relint}
 \relint \left( F_i \cap K \right) \subseteq \left( F_i \cap K \right) \setminus \bigcup_{j \neq i} \left( F_j \cap F_i \cap K \right),
\end{equation}
from which it immediately follows that $\relint \left( F_i \cap K \right) \cap \relint \left( F_j \cap K \right) = \emptyset$ for $i \neq j$. To show \eqref{eq:relint}, let the two open halfspaces on either side of the hyperplane $F_j$ be denoted by $F_j^+$ and $F_j^-$. W.l.o.g.\ assume that $K \cap F_j^- = \emptyset$. Since $\left( F_i \cap K \right) \nsubseteq \left( F_j \cap K \right)$, we must have $\left( F_i \cap K \right) \cap F_j^+ \neq \emptyset$. Let $x \in \left( F_i \cap K \right) \cap F_j^+$ and let $y \in F_j \cap F_i \cap K$. Since $F_i \cap K$ is convex, the interval from $x$ to $y$ is contained in $F_i \cap K$, but because $\left( F_i \cap K \right) \cap F_j^- = \emptyset$, no points on this line past the point $y$ can be in $F_i \cap K$; hence $y \notin \relint \left( F_i \cap K \right)$.

Since $F_i \cap K$ is a basic convex cone, $\relint \left( F_i \cap K \right)$ is a relatively open polyhedral convex cone by induction. If $F_i \cap K = \aff \left( F_i \cap K \right)$ then $\relint \left( F_i \cap K \right) = F_i \cap K$. If not, then denote by $F_{i,1}, \dots, F_{i,\ell_i}$ the hyperplanes in $\aff \left( F_i \cap K \right)$ corresponding to the finitely many faces of the polyhedron $\cl \left( F_i \cap K \right)$. By induction, the convex cones $F_{i,j} \cap F_i \cap K$, $j = 1, \dots, \ell_i$, are elementary convex cones, and we can write
\[
 K \setminus K^0 = \left( \dot{\cup}_{i=1}^{\ell} \relint \left( F_i \cap K \right) \right) \dot{\bigcup} \left( \cup_{i=1}^{\ell} \cup_{j=1}^{\ell_i} \left( F_{i,j} \cap F_i \cap K \right) \right).
\]
What remains to be shown is that $\cup_{i=1}^{\ell} \cup_{j=1}^{\ell_i} \left( F_{i,j} \cap F_i \cap K \right)$ can be written as a finite disjoint union of relatively open polyhedral convex cones; this follows by iterating the previous argument.
\end{proof}

Let us now show that $\R_{\geq 0}^d$ is a basic convex cone in $\R^d$. For $i = 1, \dots, d$, define the closed halfspace $H_i^{\geq 0} = \left\{ x \in \R^d : x_i \geq 0 \right\}$ and its complement $H_i^{<0} = \left\{ x \in \R^d : x_i < 0 \right\}$, and from these define the convex cones
\[
 K_i = H_1^{\geq 0} \cap \dots \cap H_{i-1}^{\geq 0} \cap H_i^{< 0}, \qquad i = 1, \dots, d.
\]
Then we can write $\R^{d}$ as the disjoint union of the convex cones $\R_{\geq 0}^d$ and $K_1, \dots, K_d$, showing that indeed $\R_{\geq 0}^d$ is a basic convex cone. This implies that if we can write $\R_{\geq 0}^d$ as the disjoint union of the convex cones $C_1, \dots, C_r$, then each $C_i$ is a basic convex cone, and hence, by Theorem~\ref{thm:kemperman}, an elementary convex cone.

Now let us turn to the claim in the proof of Lemma~\ref{lem:equiv}. In Lemma~\ref{lem:equiv}, we write $\Q_{\geq 0}^{m!} \setminus \left\{ 0 \right\}$ as the disjoint union of finitely many $\Q$-convex cones: $\Q_{\geq 0}^{m!} \setminus \left\{ 0 \right\} = C_0 \dot{\cup} C_1 \dot{\cup} \dots \dot{\cup} C_r$. For $i = 0, \dots, r$, let $\tilde{C}_i = \cvx \left( C_i \right)$. It is known (see, e.g.,~\cite{young1975social}) that $C_i = \Q^{m!} \cap \tilde{C}_i$. The $\tilde{C}_i$ are therefore disjoint convex cones which satisfy
\begin{equation}\label{eq:almost_bcc1}
 \tilde{C}_0 \dot{\cup} \tilde{C}_1 \dot{\cup} \dots \dot{\cup} \tilde{C}_r \subseteq \R_{\geq 0}^{m!}
\end{equation}
and
\begin{equation}\label{eq:almost_bcc2}
 \cl \left( \tilde{C}_0 \right) \cup \cl \left( \tilde{C}_1 \right) \cup \dots \cup \cl \left( \tilde{C}_r \right) = \R_{\geq 0}^{m!}.
\end{equation}
Our goal is to show that each $\tilde{C}_i$ is an elementary convex cone. Conditions \eqref{eq:almost_bcc2} and \eqref{eq:almost_bcc1} are very similar to the definition of a basic convex cone; in this spirit let us introduce the following definition.
\begin{definition}[Basic convex cone up to closure]
 Let $K_0$ be a convex cone in a finite dimensional linear space $X$. We say that $K_0$ is a \emph{basic convex cone up to closure} (in $X$) if there exist disjoint convex cones $K_1, \dots, K_r$ such that
\[
 K_0 \dot{\cup} K_1 \dot{\cup} \dots \dot{\cup} K_r \subseteq X
\]
and
\[
 \cl \left( K_0 \right) \cup \cl \left( K_1 \right) \cup \dots \cup \cl \left( K_r \right) = X.
\]
\end{definition}
Since $\R_{\geq 0}^d$ is a basic convex cone, the $\tilde{C}_i$ above are basic convex cones up to closure.

In fact, every basic convex cone up to closure is an elementary convex cone; the proof is exactly the same as the one shown above for the ``only if'' direction of Theorem~\ref{thm:kemperman}, one just needs to replace ``basic convex cone'' with ``basic convex cone up to closure'' everywhere in the proof, and make the appropriate changes. Moreover, the other direction of Theorem~\ref{thm:kemperman} implies that actually every basic convex cone up to closure is a basic convex cone.

Hence the $\tilde{C}_i$ are elementary convex cones, which is what we need in Lemma~\ref{lem:equiv}.

\section{Most voting rules are hyperplane rules: examples}\label{sec:app-examples} 

 In the following we show that all positional scoring rules, instant-runoff voting, Coombs' method, contingent vote, the Kem\'eny-Young method, Bucklin voting, Nanson's method, Baldwin's method, and Copeland's method are all hyperplane rules.

\begin{itemize}
 \item {\bf Positional scoring rules.} Let $w \in \R^m$ be a weight vector. Given a ranking profile vector $\sigma$, the (normalized) \emph{score} of candidate $a \in \left[m\right]$ is $s_a = \frac{1}{n}\sum_{i=1}^n w \left( \sigma_i^{-1} \left( a \right) \right)$. The \emph{positional scoring rule} associated to the weight vector $w$ elects the candidate who has the highest score. (In case of a tie, there is some tie-breaking rule, but we do not care about this here.) We denote such a SCF on $n$ voters by $f_n^w$. Examples include plurality (with weight vector $w = \left( 1, 0, 0, \dots, 0 \right)$), Borda count (with weight vector $w = \left( m-1, m-2, \dots, 0 \right)$) and veto (with weight vector $w = \left( 1, 1, \dots, 1, 0 \right)$).

 To a sequence of SCFs $\left\{f_n^w\right\}_{n\geq 1}$ we can associate a function $F^w : \Delta^{m!} \to \left[m\right]$ in the following way. For a candidate $a \in \left[ m \right]$ and $x \in \Delta^{m!}$, define the (normalized) score $s_a \left( x \right) = \sum_{\pi \in S_m} x_{\pi} w \left( \pi^{-1} \left( a \right) \right)$, and let
\[
 F^{w} \left( x \right) := \argmax_{a \in \left[m\right]} s_a \left( x \right),
\]
if this $\argmax$ is unique, and if it is not unique, then there is some tie-breaking rule. This construction guarantees that $f_n^w = F^{w}|_{D_n}$. For candidates $a \neq b$, define
\[
 H_{a,b} := \left\{ x \in \Delta^{m!} : s_a \left( x \right) = s_b \left( x \right) \right\},
\]
which is an affine hyperplane of the probability simplex $\Delta^{m!}$. Clearly the boundary $B^w$ is contained in the union of $\binom{m}{2}$ such affine hyperplanes:
\[
 B^w \subseteq \bigcup_{a \neq b \in \left[m \right]} H_{a,b}.
\]

 \item {\bf Instant-runoff voting.} If a candidate receives absolute majority of first preference votes, then that candidate wins. If no candidate receives an absolute majority, then the candidate with fewest top votes is eliminated. In the next round the votes are counted again, with each ballot counted as one vote for the advancing candidate who is ranked highest on that ballot. This is repeated until the winning candidate receives a majority of the vote against the remaining candidates.

 The boundary corresponds to two kinds of situations: either (1) there is a tie at the top at the end, when only two candidates remain; or (2) there is a tie for eliminating a candidate at the end of one of the rounds. Technically situation (1) is also contained in situation (2), since at the very end one can view choosing a winner as eliminating the second placed candidate. One can see that if candidates $a$ and $b$ are tied for elimination after candidates $C \subseteq \left[m\right] \setminus \left\{a,b\right\}$ (where $C = \emptyset$ is allowed) have been eliminated, then necessarily
  \[
   \sum_{C' \subseteq C} \sum_{\substack{\left\{ \pi \left( 1 \right), \dots, \pi \left( \left| C' \right| \right) \right\} = C',\\ \pi \left( \left| C' \right| + 1 \right) = a}} x_{\pi} = \sum_{C' \subseteq C} \sum_{\substack{\left\{ \pi \left( 1 \right), \dots, \pi \left( \left| C' \right| \right) \right\} = C',\\ \pi \left( \left| C' \right| + 1 \right) = b}} x_{\pi}.
  \]
 Consequently the boundary $B$ is contained in the union of at most $m^2 2^m$ affine hyperplanes:
  \[
   B \subseteq \bigcup_{a \neq b} \bigcup_{C \subseteq \left[ m \right] \setminus \left\{a,b\right\}} \left\{ x \in \Delta^{m!} : \sum_{C' \subseteq C} \sum_{\substack{\left\{ \pi \left( 1 \right), \dots, \pi \left( \left| C' \right| \right) \right\} = C',\\ \pi \left( \left| C' \right| + 1 \right) = a}} x_{\pi} = \sum_{C' \subseteq C} \sum_{\substack{\left\{ \pi \left( 1 \right), \dots, \pi \left( \left| C' \right| \right) \right\} = C',\\ \pi \left( \left| C' \right| + 1 \right) = b}} x_{\pi} \right\}.
  \]

 \item {\bf Coombs' method.} This is similar to IRV, but the elimination rule is different. If a candidate receives absolute majority of first preference votes, then that candidate wins. If no candidate receives an absolute majority, then the candidate who is ranked last by the most voters is eliminated. In the next round the votes are counted again, with each ballot counted as one vote for the advancing candidate who is ranked highest on that ballot. This is repeated until the winning candidate receives a majority of the vote against the remaining candidates.

 The boundary corresponds to two kinds of situations: either (1) there is a tie at the top at the end, when only two candidates remain; or (2) there is a tie for eliminating a candidate at the end of one of the rounds. Technically situation (1) is also contained in situation (2), since at the very end one can view choosing a winner as eliminating the second placed candidate. One can see that if candidates $a$ and $b$ are tied for elimination after candidates $C \subseteq \left[m\right] \setminus \left\{a,b\right\}$ (where $C = \emptyset$ is allowed) have been eliminated, then necessarily
  \[
   \sum_{C' \subseteq C} \sum_{\substack{\left\{ \pi \left( m \right), \dots, \pi \left( m - \left| C' \right| + 1 \right) \right\} = C',\\ \pi \left( m - \left| C' \right| \right) = a}} x_{\pi} = \sum_{C' \subseteq C} \sum_{\substack{\left\{ \pi \left( m \right), \dots, \pi \left( m - \left| C' \right| + 1 \right) \right\} = C',\\ \pi \left( m - \left| C' \right| \right) = b}} x_{\pi}.
  \]
 Consequently the boundary $B$ is contained in the union of at most $m^2 2^m$ affine hyperplanes:
{\small
  \[
   B \subseteq \bigcup_{a \neq b} \bigcup_{C \subseteq \left[ m \right] \setminus \left\{a,b\right\}} \left\{ x \in \Delta^{m!} : \sum_{C' \subseteq C} \sum_{\substack{\left\{ \pi \left( m \right), \dots, \pi \left( m - \left| C' \right| + 1 \right) \right\} = C',\\ \pi \left( m - \left| C' \right| \right) = a}} x_{\pi} = \sum_{C' \subseteq C} \sum_{\substack{\left\{ \pi \left( m \right), \dots, \pi \left( m - \left| C' \right| + 1 \right) \right\} = C',\\ \pi \left( m - \left| C' \right| \right) = b}} x_{\pi} \right\}.
  \]
}
 \item {\bf Contingent vote.} This is also similar to IRV, except here all but two candidates get eliminated after the first round. If a candidate receives absolute majority of first preference votes, then he/she wins. If no candidate receives an absolute majority, then all but the top two leading candidates are eliminated and there is a second count, where the votes of those who supported an eliminated candidate are redistributed among the two remaining candidates. The candidate who then achieves absolute majority wins.

Here the boundary $B$ corresponds to two kinds of situations: either (1) there are two distinct top candidates, and when the votes of the voters who voted for other candidates are redistributed, then the two top candidates are in a dead heat; or (2) there are two or more candidates who receive an equal number of votes in the first round. Both of these situations can be described as subsets of affine hyperplanes, and so $B$ is contained in the union of at most $m \left( m - 1 \right)$ affine hyperplanes:
  \begin{align*}
    B &\subseteq \bigcup_{a \neq b} \left\{ x \in \Delta^{m!} : \sum_{\pi : \pi \left( 1 \right) = a} x_{\pi}  + \sum_{\pi : \pi \left( 1 \right) \notin \left\{a,b \right\}, a \stackrel{\pi}{>} b} x_{\pi} = \sum_{\pi : \pi \left( 1 \right) = b} x_{\pi}  + \sum_{\pi : \pi \left( 1 \right) \notin \left\{a,b \right\}, b \stackrel{\pi}{>} a} x_{\pi}  \right\}\\
    &\cup \bigcup_{a \neq b} \left\{ x \in \Delta^{m!} : \sum_{\pi : \pi \left( 1 \right) = a} x_{\pi} = \sum_{\pi : \pi \left( 1 \right) = b} x_{\pi} \right\}.
  \end{align*}
 
 \item {\bf Kem\'eny-Young method.} Denote by $K$ the Kendall tau distance, which is a metric on permutations which counts the number of pairwise disagreements between the two permutations, i.e.,
  \[
   K \left( \tau_1, \tau_2 \right) = \sum_{\left\{a,b\right\}} \mathbf{1} \left[ a \text{ and } b \text{ are in the opposite order in } \tau_1 \text{ and } \tau_2 \right],
  \]
 where the sum is over all unordered pairs of distinct candidates. Given a ranking profile $\sigma^n$, the Kem\'eny-Young method selects the ranking which minimizes the sum of Kendall tau distances from the votes:
  \[
   \tau = \argmin \sum_{i=1}^{n} K \left( \sigma_i, \tau \right),
  \]
 and then the winner of the election is declared to be $\tau \left( 1 \right)$. For us it will be convenient to write $\tau$ as
  \[
   \tau = \argmin \sum_{\pi} x_{\pi} \left( \sigma^n \right) K \left( \pi, \tau \right).
  \]

 Here if we are on the boundary $B$ then there must exist two rankings $\tau_1$ and $\tau_2$ such that $\tau_1 \left( 1 \right) \neq \tau_2 \left( 1 \right)$ and $\sum_{\pi} x_{\pi} K \left( \pi, \tau_1 \right) = \sum_{\pi} x_{\pi} K \left( \pi, \tau_2 \right)$. Thus $B$ is contained in the union of at most $\left( m! \right)^2$ affine hyperplanes:
  \[
   B \subseteq \bigcup_{\tau_1 \neq \tau_2} \left\{ x \in \Delta^{m!} : \sum_{\pi} x_{\pi} K \left( \pi, \tau_1 \right) = \sum_{\pi} x_{\pi} K \left( \pi, \tau_2 \right) \right\}.
  \]

 \item {\bf Bucklin voting.} First every candidate gets a point from all the voters who ranked them at the top. If there is a candidate who has a majority (i.e., more than $n/2$ points), then that candidate wins. If not, then every candidate gets a point from all the voters who ranked them second. If there is a candidate who has more than $n/2$ points after this, then the candidate with the most points wins (there might be multiple candidates with more than $n/2$ points after a given round). This process is iterated until there is a candidate with more than $n/2$ points.

 Here a point on the boundary $B$ corresponds to a situation where some pair of candidates have the same number of points after some number of rounds. Therefore $B$ is contained in the union of at most $m^2 \left( m - 1 \right) / 2$ affine hyperplanes:
  \[
    B \subseteq \bigcup_{a \neq b} \bigcup_{k=1}^m \left\{ x \in \Delta^{m!} : \sum_{i=1}^k \sum_{\pi : \pi \left( i \right) = a} x_{\pi} = \sum_{i=1}^k \sum_{\pi : \pi \left( i \right) = b} x_{\pi}  \right\}.
  \]

 \item {\bf Nanson's method.} This is Borda count combined with a variation of the instant-runoff voting procedure. First, the Borda scores of all candidates are computed, and then those candidates with Borda score no greater than the average Borda score are eliminated. Then the Borda scores of each remaining candidate are recomputed, as if the eliminated candidates were not on the ballot. This is repeated until there is a final candidate left.

 The boundary corresponds to situations when a candidate's Borda score exactly equals the average score after some candidates have been eliminated. For $C \subseteq \left[ m \right]$, denote by $s_{a,C} \left( x \right)$ the score of candidate $a$ after exactly the candidates in $C$ have been eliminated ($s_{a,C} \left( x \right)$ is a linear function of $\left\{ x_{\pi} \right\}_{\pi \in S_m}$), and denote by $\bar{s}_{C} \left( x \right)$ the average score of remaining candidates after exactly the candidates in $C$ have been eliminated. 
 The boundary $B$ is contained in the union of at most $m 2^m$ affine hyperplanes:
  \[
   B \subseteq \bigcup_{a \in \left[m \right]} \bigcup_{C \subseteq \left[ m \right] \setminus \left\{a \right\}} \left\{ x \in \Delta^{m!} : s_{a,C} \left( x \right) = \bar{s}_C \left( x \right) \right\}.
  \]

 \item {\bf Baldwin's method.} This is essentially Borda count combined with the instant-runoff voting procedure. First, the Borda scores of all candidates are computed, and then the candidate with the lowest score is eliminated. Then the Borda scores of each remaining candidate are recomputed, as if the eliminated candidate were not on the ballot. This is repeated until there is a final candidate left.

 The boundary corresponds to ties for eliminating a candidate at the end of one of the rounds. Borrow the notation $s_{a,C} \left( x \right)$ from the previous example. The boundary $B$ is thus contained in the union of at most $m^2 2^m$ affine hyperplanes:
  \[
   B \subseteq \bigcup_{a \neq b} \bigcup_{C \subseteq \left[ m \right] \setminus \left\{a,b\right\}} \left\{ x \in \Delta^{m!} : s_{a,C} \left( x \right) = s_{b,C} \left( x \right) \right\}.
  \]

 \item {\bf Copeland's method.} This is a pairwise aggregation method: every candidate gets 1 point for each other candidate it beats in a pairwise majority election, and 1/2 a point for each candidate it ties with in a pairwise majority election. The winner is the candidate who receives the most points. This method corresponds to cutting the simplex $\Delta^{m!}$ up into finitely many regions via $\binom{m}{2}$ affine hyperplanes, and in each region the winner is the candidate with the most points.

 While in the previous examples tie-breaking rules were not an issue, here it does become important. We do not care about tie-breaking rules when we are on an affine hyperplane where two candidates tie each other in a pairwise majority election. However, there are open regions in the intersection of halfspaces defined by the affine hyperplanes where candidates are tied at the top with having the same scores. In this case, in order for Copeland to be a hyperplane rule, we need to break ties in favor of the same candidate for the whole region. (This is also how Xia and Conitzer break ties for Copeland's method in~\cite{xia2008generalized}.)

 Using this tie-breaking rule Copeland's method is indeed a hyperplane rule, since the boundary is contained in the union of at most $\binom{m}{2}$ affine hyperplanes:
  \[
   B \subseteq \bigcup_{a \neq b} \left\{ x \in \Delta^{m!} : \sum_{\pi : a \stackrel{\pi}{>} b} x_{\pi} = \sum_{\pi : b \stackrel{\pi}{>} a} x_{\pi}  \right\}.
  \]
%
\end{itemize}

\end{document}